\documentclass[12pt, draftclsnofoot, onecolumn]{IEEEtran}
\usepackage{amsthm}
\usepackage{amsmath}
\usepackage{algorithmic}
\usepackage[ruled,linesnumbered]{algorithm2e}
\usepackage{mathrsfs}
\usepackage{cite}
\usepackage{bm,epsfig,amsthm,url}
\usepackage{indentfirst}
\usepackage{amssymb}
\usepackage{epstopdf}
\usepackage[caption=false,font=footnotesize]{subfig}
\usepackage{xcolor}
\usepackage{mathtools}

\theoremstyle{definition}
\newtheorem{lemma}{Lemma}
\newtheorem{theorem}{Theorem}

\newtheorem{proposition}{Proposition}

\newtheorem{remark}{Remark}

\begin{document}
	
	\title{UAV-Assisted Multi-Cluster Over-the-Air Computation}
	\author{Min~Fu,~\IEEEmembership{Student Member,~IEEE},~Yong~Zhou,~\IEEEmembership{Member,~IEEE},  Yuanming~Shi,~\IEEEmembership{Senior Member,~IEEE}, Chunxiao~Jiang,~\IEEEmembership{Senior Member,~IEEE}, and Wei~Zhang,~\IEEEmembership{Fellow,~IEEE}
    	\thanks{M. Fu, Y. Zhou,  and Y. Shi are with School of Information Science and Technology, ShanghaiTech University, Shanghai 201210, China (e-mail: \{fumin, zhouyong, shiym\}@shanghaitech.edu.cn).}    
    	\thanks{C. Jiang is with the Tsinghua Space Center and the Beijing National Research Center for Information Science and Technology, Tsinghua University, Beijing 100084, China (e-mail: jchx@tsinghua.edu.cn).}  	   	
    	\thanks{W. Zhang is with the School of Electrical Engineering and Telecommunications, The University of New South Wales, Sydney, NSW 2052, Australia
    	(e-mail: w.zhang@unsw.edu.au).}
    	
}	
\maketitle
%\setlength\abovedisplayskip{2pt}
%\setlength\belowdisplayskip{2pt}
%\vspace{-3mm}
\begin{abstract}
In this paper, we study unmanned aerial vehicles (UAVs) assisted wireless data aggregation (WDA) in multi-cluster networks, where multiple UAVs simultaneously perform different WDA tasks via over-the-air computation (AirComp) without terrestrial base stations.
This work focuses on maximizing the minimum amount of WDA tasks performed among all clusters by optimizing the UAV’s trajectory and transceiver design as well as cluster scheduling and association, while considering the WDA accuracy requirement.
Such a joint design is critical for interference management in multi-cluster AirComp networks, via enhancing the signal quality between each UAV and its associated cluster for signal alignment and meanwhile reducing the inter-cluster interference between each UAV and its non-associated clusters.
Although it is generally challenging to optimally solve the formulated non-convex mixed-integer nonlinear programming, an efficient iterative algorithm as a compromise approach is developed by exploiting bisection and block coordinate descent methods, yielding an optimal transceiver solution in each iteration.
The optimal binary variables and a suboptimal trajectory are obtained by using the dual method and successive convex approximation, respectively.
Simulations show the considerable performance gains of the proposed design over benchmarks and the superiority of deploying multiple UAVs in increasing the number of performed tasks while reducing access delays.
\end{abstract}
	
\begin{IEEEkeywords}
	Over-the-air computation, UAV communications, wireless data aggregation, multi-cluster cooperation, interference management.
\end{IEEEkeywords}
	
\section{Introduction}	\label{Section:introduction}

Machine-type communication (MTC) is one of the disruptive technologies promised by 5G and beyond wireless networks \cite{letaief2022edge}.
Therein, it is crucial to collect and leverage Big Data effectively for decision-making to automate various intelligent applications.
However, collecting data generated by an enormous number of devices in the future Internet of Things (IoT) networks is critically challenging due to limited spectrum resource \cite{agiwal2016next}, \cite{Shi2020Communication}.
Meanwhile, many emerging IoT applications (e.g., environmental monitoring) only aim to collect a particular function of these massive data rather than to reconstruct each individual data, which is referred to as wireless data aggregation (WDA) \cite{Zhu2020WDA}.
To meet these demands, over-the-air computation (AirComp) has recently been considered as an attractive technique to enable fast WDA among massive devices by seamlessly integrating communication and computation processes \cite{Golden2013AirComp}. 
The principle of AirComp is to exploit the waveform/signal superposition property of multiple-access channels (MAC) such that an edge server directly receives a function of concurrently transmitted data.
This results in low transmission delays regardless of the amount of devices, and makes AirComp particularly appealing to data-intensive and/or latency-critical applications such as consensus control\cite{Molinari2021Consensus}, distributed sensing \cite{Liu2019AirCompTWC}, and distributed machine learning\cite{Yang2020FL},\cite{Wang2021FLRISTWC}.

So far, AirComp has been studied from various aspects in single-cell networks, such as single-input-single-output (SISO) AirComp \cite{Liu2019AirCompTWC, Cao2020Optimized}, multiple-input-single-output (MISO) AirComp \cite{Chen2018UniformForcing, Fang2021AirComp}, and multiple-input-multiple-output (MIMO) AirComp \cite{Zhu2019Aircompmobility}.
In single-cell networks, to achieve accurate computing, AirComp requires the phase and magnitude of all signals to be aligned at the receiver side. 
However, channel heterogeneity across devices makes signal alignment challenging.
To cope with this issue, different transceiver designs have been proposed to compensate for the non-uniform channel fading and suppress the noise.
Specifically, for SISO AirComp, the authors in \cite{Liu2019AirCompTWC, Cao2020Optimized} proposed the optimal transmit power control and receive normalizing factor design.
For multi-antenna AirComp systems, beamforming vectors at receiver and/or transmitter were designed to minimize the computation error in \cite{Chen2018UniformForcing,Zhu2019Aircompmobility}.
However, the AirComp performance may deteriorate when one or more power-constrained devices are in deep fading. 
To mitigate the communication bottleneck, the authors in \cite{Fang2021AirComp} employed a promising reconfigurable intelligent surface (RIS) \cite{Yuan2021RIS,Fu2021Intelligent} to jointly design the passive beamforming at the RIS and transceiver.
In addition, we in \cite{Fu2021UAV}, \cite{Fu2021UAVAirComp} proposed to deploy an unmanned aerial vehicle (UAV) as a mobile server and to exploit its mobility to avoid any device being in deep fade, thereby enhancing the performance of AirComp.

Meanwhile, the IoT networks generally involve different WDA tasks, each of which is characterized by their applications (e.g., classification FL task \cite{Wen2022joint}, regression FL task, and sensing tasks \cite{Lan2020simultaneous}), data types (e.g., model parameters in machine learning and velocity in connected car platooning applications), and computing functions (e.g., sum and mean).
Hence, researchers recently advocated the study of AirComp in a multi-cluster network to simultaneously complete multiple WDA tasks, which is referred to as multi-cluster AirComp \cite{Zhu2020WDA}.
In multi-cluster AirComp, besides signal alignment and noise suppression, inter-cluster interference management arises as a new challenge.
Different from the orthogonal multiple access (OMA) based multi-cluster networks \cite{Shi2014GroupSparse}, each server in the multi-cluster AirComp network not only harnesses intra-cluster interference for function computation, but also suppresses inter-cluster interference for computation error reduction.
As an initial study, the authors in \cite{Lan2020simultaneous} proposed a signal-and-interference alignment scheme to simultaneously eliminate inter-cell interference and align intra-cell signals in a two-cell MIMO AirComp network.
Subsequently, the authors in \cite{Cao2021multiCellAirCompTWC} studied the weighted sum mean-square error (MSE) minimization problem in multi-cell SISO AirComp networks by optimizing the transmit power of devices.
Results in \cite{Cao2021multiCellAirCompTWC} showed that, the performance of AirComp is limited by the inter-cell interference since the transmit power of the device needs to balance the tradeoff between combating inter-cell interference and enhancing signal alignment within the cell.
These results are established when the static terrestrial BS is available.
However,  in remote and under-developed areas, the terrestrial BSs are usually sparsely deployed or not available. 
In these harsh circumstances, it is critical to deploy more flexible BSs to unleash the potential of multi-cluster AirComp.

As a parallel but complementary study, in this paper, we investigate a novel multi-cluster AirComp framework with multiple UAVs dispatched as flying BSs to cooperatively perform diverse AirComp tasks, where no terrestrial BS is available.
In fact, multiple UAVs have been deployed as BSs to assist the terrestrial networks for rate-oriented communications \cite{Wu2018MultiUAV}, \cite{Zhan2019timeTWC}, \cite{Shen2020mutiUAV}.
For example, the authors in \cite{Wu2018MultiUAV} studied multi-UAV cooperative communications to achieve higher data rate and lower access delay.
The authors in \cite{Zhan2019timeTWC} considered a multi-UAV enabled wireless network for data collection. 
This motivates us to study the use of multiple UAVs in multi-cluster AirComp by taking into account UAV trajectory planning as well as cluster scheduling and association for interference management, which has the following advantages.
First, multi-UAV cooperation exploits spectrum sharing to allow multiple clusters to be simultaneously served, thereby increasing the amount of performed task within the given time duration.
Second, scheduling the clusters that are far from each other can avoid strong co-channel interference and thus improve computation accuracy.
Furthermore, the joint design not only shortens the communication distances between the UAV and its associated cluster to enhance intra-cluster signal alignment but also enlarges the communication distances between the UAV and its non-associated clusters to rein in the inter-cluster interference.
This joint design is crucial for interference management in multi-cluster AirComp networks, but has not yet been studied in the literature to the best of the authors’ knowledge.

A widely adopted performance metric for quantifying the AirComp computation error is MSE between the estimated function value and the ground truth\cite{Yang2020FL, Zhu2019Aircompmobility, Cao2020Optimized}.
To promote fairness among clusters, under the given target MSE requirements, we aim to maximize the minimum amount of performed WDA tasks among all clusters by jointly designing cluster scheduling and association, UAVs' trajectories, and transceiver design in a given time duration.
However, such a non-convex max-min fairness problem presents unique challenges due to their discontinuous objective functions (as a result of the binary cluster scheduling and association variables), non-convex MSE constraints (because of the coupling between all optimization variables), and non-convex trajectory design.
It is generally challenging to optimally solve such a mixed-integer non-convex optimization problem.

\subsection{Contributions}
The main contributions of this paper are summarized as follows.
\begin{itemize}
	\item 
	This paper is one of the early attempts to study the multi-UAV enabled AirComp in multi-cluster wireless networks, where multiple UAVs are deployed to cooperatively perform different AirComp tasks.
	To ensure fairness among clusters, we aim to maximize the minimum amount of performed WDA tasks among all clusters by jointly optimizing scheduling and association for clusters, UAVs' trajectory planning, and transceiver design, taking into account the target accuracy of AirComp, practical constraints on UAVs, inter-cluster interference, as well as the total power budgets at devices.
	
	\item  
	To render the resulting non-convex mixed-integer nonlinear programming (MINLP) tractable, by leveraging the bisection method, the original problem reduces to a sequence of the minimum ratio maximization problems, which enables the development of an iterative algorithm.
	After adopting block coordinate descent (BCD) \cite{xu2013block}, besides both convex normalizing factors and power optimization subproblems being optimally solved, the optimal cluster scheduling and association are also obtained by applying the low-complexity Lagrange duality method.
	For the non-convex UAV trajectory optimization problem, a suboptimal solution is obtained by using the successive convex approximation (SCA) method.	
	
	\item 
	Simulation results are presented to show the effectiveness and superiority of the proposed design and developed algorithm.
	In multi-cluster AirComp with a single UAV, the performance achieved by the joint design outperforms other benchmarks and can reach the upper bound.
	It is also shown that the max-min task amount of the considered UAV network increases with the mission duration, revealing a performance-access delay tradeoff in multi-cluster AirComp.
	Compared to the single-UAV case, the use of multiple UAVs with effective cooperative interference management can considerably increase the amount of tasks performed by each cluster while reducing access delays.
	
\end{itemize}

\subsection{Organization}
The rest of this paper is organized as follows. 
We presents the system model and the problem formulation for multi-cluster AirComp assisted by UAVs in Section \ref{Section:model}. 
 We develop an iterative algorithm yielding high-quality solutions to solve the formulated problem based on the Bisection method, BCD, and SCA in Section \ref{Section:algorithm}.
In addition,  numerical results are presented in Section \ref{Section:simulation} to evaluate the performance of the proposed design. 
In Section \ref{Section:conclusion}, we draw the conclusions.

\textit{Notations}: Scalars, column vectors, and matrices are written in italic letters, boldfaced lower-case letters, and boldfaced upper-case letters respectively, e.g., $a$, $\mathbf{a}$, $\mathbf{A}$. $\mathbb{R}^{M \times N}$ denotes the space of a real-valued matrix with with $M$ rows and $N$ columns. 
	 $\left\| \mathbf{a} \right\|_2$ denotes the Euclidean norm of vector $\mathbf{a}$ and $\mathbf{a}^T$ represents its transpose. $|\mathcal{S}|$ denotes the cardinality of the set $\mathcal{S}$.
	 
\section{System Model And Problem Formulation}\label{Section:model}	
\begin{figure}[t]
	\centering	\includegraphics[scale = 0.3]{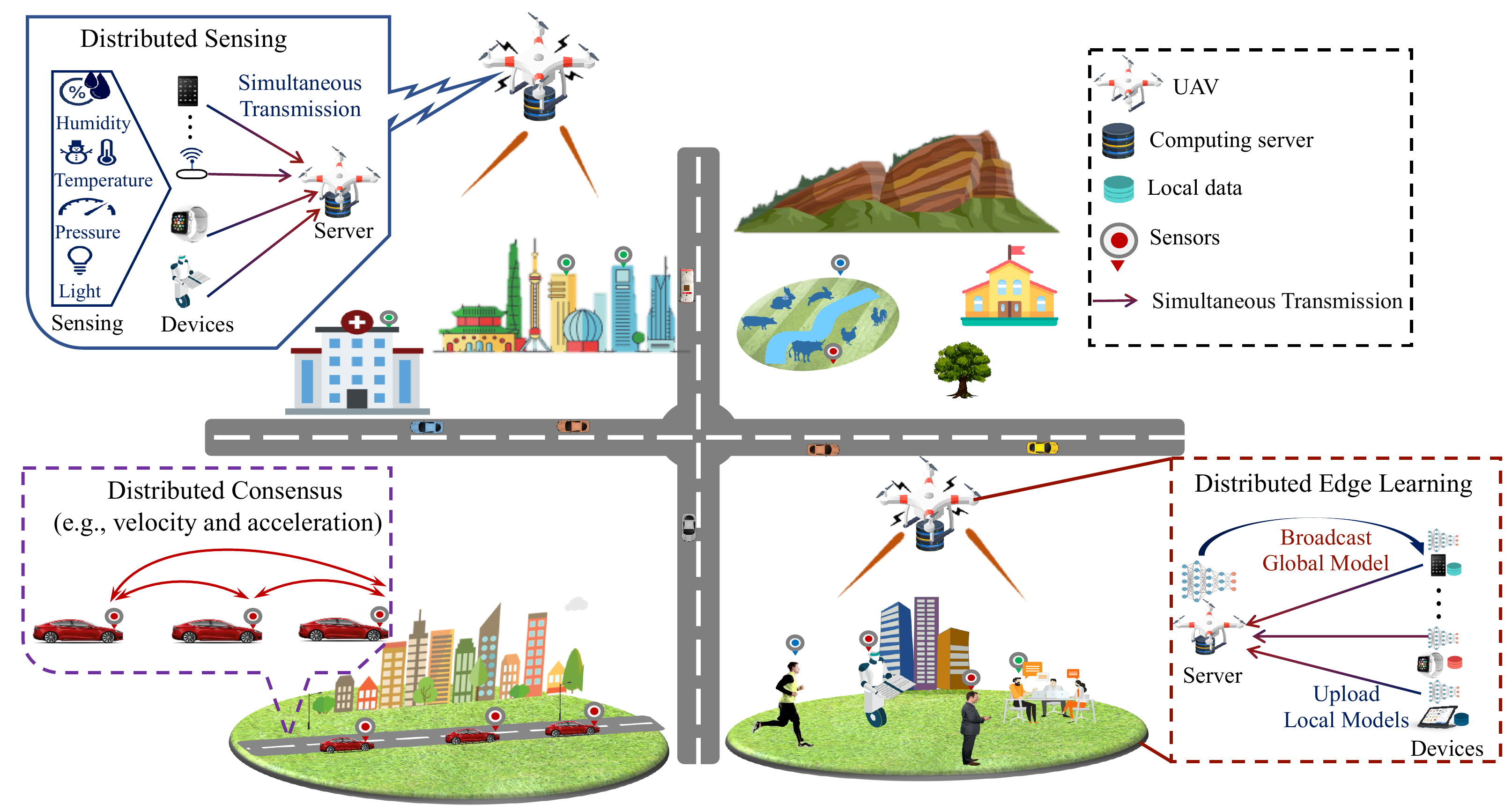}
	\vspace{-3mm}
	\caption{Illustration of multi-UAV assisted AirComp in multi-cluster wireless networks.}
	\label{model}
	\vspace{-8mm}
\end{figure}	
\subsection{System Model}
As shown in Fig. \ref{model}, we consider a multi-cluster wireless network, where each cluster consists of multiple ground devices and  multiple UAVs\footnote{We consider the scenario that the terrestrial BS is not available, e.g., in wild areas. Under this circumstance, the UAVs are adopted as an alternative to provide wireless services for ground devices \cite{Wu2018MultiUAV}. }
are deployed to support multiple WDA tasks using AirComp.
Due to the limited physical size, we assume that UAVs and ground devices are all equipped with one antenna.
Let $\mathcal{M} \triangleq \{1,\ldots, M\}$ and  $\mathcal{L} \triangleq \{1,2,\ldots,L\}$ denote the sets of $M$ UAVs and $L$ clusters, respectively.
We denote $\mathcal{K}_l \triangleq \{1,\ldots, K_l\}$ 
as the set of $K_l$ ground devices in cluster $l$ with $\mathcal{K}_i \cap \mathcal{K}_j = \emptyset, \forall i \neq j, i, j \in \mathcal{L}$, and $\mathcal{K} \triangleq \cup_{l\in \mathcal{L}} \mathcal{K}_l$ as the set of all ground devices.
Therein, the ground devices in set $\mathcal{K}_l$ need to collaboratively perform the type-$l$ WDA task, i.e., computing a specific nomographic function (e.g., arithmetic mean) with respect to their measured data, and transmit the aggregated data to one of the UAVs. 
To avoid excessive communication latency,  we adopt the AirComp technique \cite{Yang2020FL} to achieve fast  WDA by enabling concurrent transmission among multiple ground devices.
Compared to the OMA scheme, the unique feature of AirComp is the utilization of  the waveform superposition property of a MAC to harness the  intra-cluster interference for function computation, thereby reducing the communication latency.
Hereinafter, we refer to this type of  WDA task by using the AirComp technique as an AirComp task.

We assume that the AirComp tasks in different clusters are distinct, i.e., different classes of IoT applications supported by different clusters (e.g.,  distributed machine learning  \cite{Yang2020FL}, \cite{Wang2021FLRISTWC} and control consensus \cite{Molinari2021Consensus}).
Among them, their data types (e.g., model parameters in machine learning and velocity in connected car platooning applications) and the intended computing functions (e.g., sum and mean) can be different.
For UAV deployment, it is not practical to dispatch one or more UAVs for each cluster, especially when the number of clusters served is relatively large.
Thus, in this paper, we assume that  $M\leq L$ and
each UAV can select an appropriate cluster for association. 
Different UAVs can simultaneously deal with different AirComp tasks over the same frequency band, thereby reducing the access delay of ground devices in all clusters.
As a result, this leads to inter-cluster interference among $M$ clusters and calls for interference management to balance the computing errors among different clusters.

%\subsection{UAV-Enabled Multi-Cluster AirComp}
We aim to optimize the trajectories of UAVs to improve the WDA accuracy of all clusters.
This can be achieved by jointly designing the UAVs' trajectories and UAV-cluster association to shorten the communication distances between the UAV and its associated cluster to enhance intra-cluster signal alignment, and enlarge the communication distances between the UAV and its non-associated clusters to rein in the inter-cluster interference.
To facilitate trajectory design, in this paper,
time discretization technique is adopted to divide the mission duration  $T$ (s) into $N$ equal time-slots with time step-size $\delta = \frac{T}{N}$, as in most of the previous studies on UAV-enabled wireless networks \cite{Wu2018MultiUAV, Zhan2019timeTWC, Ma2021JSAC}.
The set of $N$ time slots  is denoted as $\mathcal{N} \triangleq \{1,\ldots,N\} $. 
Note that the time interval $\delta$ needs to be appropriately set, so that in each time slot, each UAV can complete one AirComp task while the changes in the communication distances between the UAVs and the ground devices are negligible.

In a three dimensional (3D) Cartesian coordinate system,
we assume that all UAVs has a fixed altitude of $H$ (m) over the horizontal plane, where $H$ can be set to the minimum altitude to ensure that  the UAVs can avoid obstacles (e.g., building or terrain) without the need for frequent ascent and descent.
We denote, at time slot $n$, the horizontal coordinate of UAV $m$ as $\bm q_m[n] \triangleq [x_m[n],y_m[n]]^{\sf T} \in \mathbb{R}^{2 \times 1}$.
Moreover, the fixed horizontal coordinate of each ground device $k \in \mathcal{K}$ is denoted as $\mathbf{w}_k = [x_k,y_k]^T \in \mathbb{R}^{2 \times 1}$, which is assumed to be known at the UAVs.
\subsubsection{Channel Model}  
A large body of studies on air-to-ground channel modeling \cite{Matolak2017channelModel}, \cite{Channel2017UAV}.
 And recent field experiments by Qualcomm \cite{Qual2017LET} has verified that when UAVs fly above a certain altitude, there is a high probability for the UAV-to-ground channel to be dominated by a line-of-sight (LoS) link. 
Therefore, the uplink channels from devices to UAVs are assumed to be dominated by LoS links in this paper.
The channel between UAV $m\in \mathcal{M}$ and device $k\in \mathcal{K}$ at time slot $n\in \mathcal{N}$ is modeled as 
 \setlength\arraycolsep{2pt}
\begin{eqnarray}\label{Eq:channel}
	h_{k,m}[n] = \sqrt{\beta_0 d^{-\gamma}_{k,m}[n] }e^{-j\theta_{k,m}[n]},
\end{eqnarray} 
where $\beta_0$ denotes the channel power gain at the reference distance of $d_0 = 1$ m, $\gamma \ge 2$ is the path loss exponent, and $\theta_{k,m}[n]$ is phase shift component. 
We assume that all devices can perfectly estimate its own phase component with its associated UAV.
And $d_{k,m}[n]$ denotes the distance from UAV $m$ to ground device $k$, which is represented as 
\begin{eqnarray}\label{Eq:distance}
	d_{k,m}[n] = \sqrt{H^2+\left\|\bm q_m[n]-\mathbf{w}_k\right\|_2^2}.
\end{eqnarray} 
Note that the Doppler effect is assumed to be perfectly compensated \cite{Mengali1997Sync}.

\subsubsection{Multi-cluster AirComp}
We define a set of binary variables $\{a_{l,m}[n] |  l \in \mathcal{L}, m\in\mathcal{M}, n\in\mathcal{N}\}$ to represent UAV-cluster scheduling and association over different time slots. 
We set $a_{l,m}[n] = 1$ if there exist uplink transmissions from ground devices in cluster $l$ to UAV $m$ at time slot $n$, and $a_{l,m}[n] = 0$ otherwise.
Therefore, $a_{l,m}[n]$ not only indicates the association status between  UAV $m$ and cluster  $l$ at slot $n$, but also determines the communication scheduling for cluster $l$ at slot $m$.
Suppose that, in each time slot,  the ground devices in each cluster can communicate with at most one UAV, and each UAV can serve at most one cluster.
In summary, the binary variables impose the following constraints.
\begin{eqnarray}
	&&\sum_{m=1}^M a_{l,m}[n] \leq 1, \forall l \in \mathcal{L}, n\in\mathcal{N}, \label{Cons:A UAV}\\
	&&\sum_{l=1}^La_{l,m}[n] \leq 1, \forall m \in \mathcal{M}, n\in\mathcal{N}, \label{Cons:A cluster}\\
	&&a_{l,m}[n]\in\{0,1\}, \forall l\in\mathcal{L},m\in\mathcal{M},n\in\mathcal{N}. \label{Cons:A binary}
\end{eqnarray}

We consider the case that each cluster aims to compute the average of distributed data measured by the ground devices \cite{Yang2020FL}, \cite{Cao2020Optimized}.
In the following,  we take cluster $l$ associated with UAV $m$ as an example.
Let $z_{k}[n] \in \mathbb{C}$ denote the measured data from device $k\in \mathcal{K}_l$ in cluster $l$ at  time slot $n$.
We denote the target function  (e.g., arithmetic mean) of $K_l$ variables as $f_l[n]:\mathbb{C}^{K_l}\rightarrow\mathbb{C}$.
Therefore, the target average function of type-$l$ data at time slot $n$ is expressed as 
\begin{eqnarray}\label{Eq: desired function}
	{f}_l[n] = \frac{1}{K_l} \sum_{k\in\mathcal{K}_l} s_k[n],
\end{eqnarray}
where   $s_k[n] \triangleq \psi_{k} (z_k[n]), \forall k\in \mathcal{K}_l$ with $\psi_{k}$ denoting the pre-processing function  at device $k\in \mathcal{K}_l$.
We assume that  $\{s_k[n]\}_{k\in\mathcal{K}_l}$ are independent and have zero mean and unit variance as in \cite{Cao2020Optimized},\cite{Liu2019AirCompTWC}, i.e., $\mathbb{E}(s_k[n]) = 0$, $\mathbb{E}(s_k[n]s_k^{\sf H}[n]) = 1$, and $\mathbb{E}[s_{i}[n]s_{j}[n]^{\sf H}] = 0 , \forall i \neq j$.

In the sequel, we describe how to estimate the desired function $f_l[n]$ in \eqref{Eq: desired function} at UAV $m$ via AirComp.
Let $\mathcal{J}(n)=\{l|\sum_{m\in\mathcal{M}}a_{l,m}[n]=1, \forall l \in \mathcal{L}\}$ denote the index set of active clusters at time slot $n$.
After all devices in set $\mathcal{K}_l$ simultaneously send their pre-processed signals $\{s_{k}[n]\}$ to UAV $m$ over the same radio channel, the received signal at UAV $m$ is written as
\begin{eqnarray}\label{Eq:received signal}
	&&\hspace{-1em} y_{m}[n] = \sum_{k\in\mathcal{K}_{l}} b_k[n]h_{k,m}[n]s_{k}[n] + \sum_{j\in\mathcal{J}(n)\backslash\{l\}}\sum_{i\in\mathcal{K}_j}h_{i,m}[n]b_i[n]s_i[n] + e_{m}[n],
\end{eqnarray}
where  $b_k[n]\in \mathbb{C}$ denotes the transmit precoding coefficient at device $k$ for channel-fading compensation at time slot $n$, and $e_m[n]$ denotes the  additive white Gaussian noise (AWGN) at UAV $m$, i.e., $e[n]\thicksim\mathcal{C}\mathcal{N}$$(0, \sigma^2)$.
The ground devices' transmissions are assumed to be synchronized \cite{Liu2019AirCompTWC, Cao2020Optimized}.
It is worth mentioning that the synchronization techniques reported in the literature can be applied in our work, such as the AirShare technique \cite{abari2015airshare} and the timing advance mechanism  used in the long term evolution systems\cite{Zhu2020WDA}.

Upon receiving signal $y_{m}[n]$, after post-processing and scaling at UAV, the estimated average function $m$ is given by
\begin{eqnarray}\label{Eq:estimated function}
	\hat{f}_{l, m}[n] =  \frac{\eta_{l, m}[n]y_{m}[n]}{K_l},
\end{eqnarray}
where factor $1/K_l$ is employed for averaging purpose and
$\eta_{l, m}[n]\in \mathbb{C}$ is a normalizing factor at UAV $m$ applied to scale received signal $y_{m}[n]$ for compensating channel fading and suppressing noise, thereby
accurately estimating the target function $f_{l}[n]$.	
It can be observed from \eqref{Eq:received signal} and \eqref{Eq:estimated function} that UAV $m$ harnesses the intra-cluster interference for computing the target function but suffers from inter-cluster interference due to the communication resource reuse among UAVs.
Each device $k$ has a total transmit power budget $P_k$, which is given by
\begin{eqnarray}\label{sum p}
	\sum_{n=1}^{N} |b_k[n]|^2 \leq  P_k, \forall k,
\end{eqnarray}
where the total power constraint at each ground device is considered for adaptive power allocation over different time slots.

\subsection{Performance Metrics}
To quantify the performance of AirComp,  the computation distortion of the ground true average function $f_l[n]$ is measured by its MSE. 
Specifically, when $a_{l,m}[m] = 1$, the corresponding instantaneous MSE of type-$l$ AirComp task at UAV $m$ at time slot $n$ is given by
\begin{eqnarray}\label{Eq:instant MSE-m}
{\sf{mse}}_{l,m,n}
&=&	\mathbb{E}[|\hat{f}_{l,m}[n] - f_l[n]|^2] =\frac{1}{K_l^2}\mathbb{E}\Big[\Big|{\eta_{l,m}[n]}{y_{m}[n]}-\sum_{k\in\mathcal{K}_l}s_k[n]\Big|^2\Big]\nonumber\\
	&=& \dfrac{1}{K_l^2} \Big[\sum_{k\in\mathcal{K}_l}\Big(\eta_{l,m}[n]b_k[n]h_{k,m}[n]\!-1\Big)^2 +\eta_{l,m}^2[n]\Big(\hat{I}_{l,m}[n] + \sigma^2\Big)\Big], 
\end{eqnarray}
where $ \hat{I}_{l,m}[n] \triangleq \sum_{j \in\mathcal{J}(n)\backslash\{l\}}\sum_{i\in\mathcal{K}_{j}}|b_i[n]|^2|h_{i,m}[n]|^2$
represents the inter-cluster interference received at UAV $m$ when it associates with cluster $l$ at time slot $n$.
The expectation in \eqref{Eq:instant MSE-m} is taken over the distributions of signals $\{s_k[n]\}$ and  noise $e_m[n]$. 
The instantaneous MSE of type-$l$ AirComp task at time slot $n$ is represented as
\begin{eqnarray} \label{Eq:instant MSE1}
	{\sf {MSE}}_{l,n} = \sum_{m=1}^Ma_{l,m}[n]{\sf mse}_{l,m,n}.
\end{eqnarray}

Before formulating the optimization problem, we first simplify the transmit precoding coefficients and binary variable design.
Note that  the binary variables in \eqref{Eq:instant MSE1} are coupled across different clusters and UAVs due to the interference term. 
The interference terms in ${\sf mse}_{l,m,n}$ are simplified as
\begin{eqnarray}
 I_{l,m}[n] \triangleq \sum_{j \in\mathcal{L}\backslash\{l\}}\sum_{i\in\mathcal{K}_{j}}|b_i[n]|^2|h_{i,m}[n]|^2.
\end{eqnarray}
This is because if there is no UAV associated with cluster $j$, i.e, $\sum_{m\in\mathcal{M}}a_{j,m}[n]=0$, then the transmission power of device $i\in\mathcal{K}_{j}$ must be zero (i.e., $b_i[n]=0$).
Furthermore, it can also be  observed from \eqref{Eq:instant MSE-m}  that the phases of $b_i[n]$ and $h_{i,m}[n]$ for  $\forall j \in\mathcal{L}\backslash\{l\}, \forall i\in\mathcal{K}_{j}$ do not affect the inter-cluster interference induced error in ${\sf{mse}}_{l,m,n}$ while  the essential condition for ${\sf{mse}}_{l,m,n}$ to reach the minimum value is the terms $\eta_{l,m}[n]b_k[n]h_{k,m}[n], \forall k\in\mathcal{K}_l$ in \eqref{Eq:instant MSE-m} must be real and non-negative.
Therefore, with $\eqref{Cons:A UAV}$, when $a_{l,m}[n] = 1$, we have
\begin{eqnarray} \label{Neq:instantaneous MSE}
{\sf {MSE}}_{l,n} = {\sf mse}_{l,m,n}\geq \dfrac{1}{K_l^2}\Big[\sum_{k\in\mathcal{K}_l}\Big(\eta_{l,m}[n]\sqrt{p_k[n]}|h_{k,m}[n]|-1\Big)^2 +\eta_{l,m}^2[n]\Big(I_{l,m}[n] + \sigma^2\Big)\Big], 
\end{eqnarray}
where the equality holds when $\eta_{l, m}[n]\in \mathbb{R}^{+}$ and
$b_k[n]\triangleq \frac{ \sqrt{p_{k}[n]} h_{k,m}^{\dagger}[n]}{ |h_{k,m}[n]|}$ with $p_k[n] \in \mathbb{R^+}$ denoting the transmit power of device $k$ at time slot $n$. please refer to \cite{Fu2021UAV} for its proof.
Meanwhile, if $\sum_{m=1}^{M}a_{l,m}[n] = 0$, we have ${\sf {MSE}}_{l,n}=0$ and $b_k[n] = 0, \forall k\in \mathcal{K}_l$.
Hence, at time slot $n$, the achievable instantaneous MSE of type-$l$ data  can be represented as  	
\begin{eqnarray}\label{Eq:instant MSE}
	{\sf {MSE}}_{l,n} &= &\sum_{m=1}^Ma_{l,m}[n]{\sf mse}_{l,m,n} \nonumber \\
	&=& \sum_{m=1}^M\frac{a_{l,m}[n]}{K_l^2} \Big[\sum_{k\in\mathcal{K}_l}\Big(\eta_{l,m}[n]|h_{k,m}[n]|\sqrt{p_k[n]}-1\Big)^2 + \eta_{l,m}^2[n]\Big(I_{l,m}[n] + \sigma^2\Big)\Big], 
\end{eqnarray}  	
which is a linear function of $a_{l,m}[n], \forall m$  but not coupled over different clusters, and is related to
real transmit power $p_k[n]$ instead of complex transmit precoding coefficient $b_k[n]$.
\begin{remark}
One can observe that ${\sf {MSE}}_{l,n}$ is composed of three components, including the signal misalignment error (i.e., $ \sum_{m=1}^M\frac{a_{l,m}[n]}{K_l^2}\sum_{k\in\mathcal{K}_l}\big(\eta_{l,m}[n]|h_{k,m}[n]|\sqrt{p_k[n]}-1\big)^2 $), the inter-cluster interference-induced error (i.e., $\sum_{m=1}^M\frac{a_{l,m}[n]}{K_l^2} \eta_{l,m}^2[n] \sum_{j \in\mathcal{L}\backslash\{l\}}\sum_{i\in\mathcal{K}_{j}}p_i[n]|h_{i,m}[n]|^2$ ), and the noise-induced error (i.e., $ \sum_{m=1}^M\frac{a_{l,m}[n]}{K_l^2}\eta_{l,m}^2[n]\sigma^2$).
The device can increase its power $p_k[n]$ to compensate for channel fading for signal alignment whereas it leads to an increased interference  for other co-channel clusters. Meanwhile, enlarging $\eta_{l,m}[n]$ can reduce the signal-misalignment error but magnify the negative impact of noise and inter-cluster interference.
Fortunately, the trajectory design with cluster scheduling and association can be exploited to balance the tradeoff between signal misalignment error reduction and co-channel interference reduction.
Specifically, a UAV moves closer to the associated cluster $l$ to construct strong desired links $|h_{k,m}[n]|, \forall k\in\mathcal{K}_l$ for signal alignment and keeps away from other scheduled clusters to construct weak co-channel interference links $|h_{i,m}[n]|^2, i\in\mathcal{K}_{j}, j\neq l$.
\end{remark}

\subsection{Problem Formulation}
In the sequel, let $\bm{A} = \{a_{l,m}[n],\forall l\in\mathcal{L},\forall m\in\mathcal{M}, \forall n \in\mathcal{N}\}$, $\bm{P} = \{p_k[n], \forall k\in\mathcal{K},\forall n\in \mathcal{N}\}$, $\bm{\eta} = \{\eta_{l,m}[n], \forall l\in\mathcal{L},\forall m\in\mathcal{M}, \forall n \in\mathcal{N}\}$, and $\bm{Q} = \{\bm q_m[n],\forall m\in\mathcal{M},\forall n\}$.
Given mission duration $T$ and certain instantaneous MSE requirements for different AirComp tasks, we study the total amount of AirComp tasks  maximization problem while ensuring fairness among clusters.
Therefore, we maximize the minimum amount of AirComp tasks among all clusters, i.e., $\min_{l}\sum_{n=1}^N\sum_{m=1}^{M} a_{l,m}[n]$, by jointly optimizing the cluster scheduling and association indicator variables $\bm{A}$, transmit power $\bm{P}$ at devices, signal normalizing factors $\bm{\eta}$ at the UAVs, and  UAVs' trajectories $\bm Q$. 
The optimization problem is formulated as
\begin{subequations}\label{Problem:original}
	\begin{eqnarray}
\mathscr{P}:\mathop{\text{maximize}}\limits_{\bm{A},\bm{P}, \bm{\eta},\bm{Q}} \mathop{\text{min}}\limits_{l} && \sum_{n=1}^N\sum_{m=1}^{M} a_{l,m}[n] \label{P} \nonumber\\
		\text{subject to } 
		&& {\sf{MSE}}_{l,n} \le \varepsilon_l, \forall n, \forall l, \label{Cons:target MSE}\\
		&& 0 \leq \sum_{n=1}^Np_k[n]\leq P_k, \forall k,\label{Cons:power}\\
		&& \bm q_m[0] = \bm q_m[N], \forall m \in\mathcal{M}, \label{Cons:Q position}\\ 
		&&\left\|\bm q_m[n] - \bm q_m[n-1]\right\|_2 \le V_{\text{max}}\delta,    \forall n\in\mathcal{N}, \forall m\in\mathcal{M}, \label{Cons:Q speed}\\
		&&\left\|\bm q_m[n] - \bm q_i[n] \right\|_2 \geq d_{\text{min}}, \forall m \neq i, i, m \in \mathcal{M}, \forall n,  \label{Cons:Q collision}\\
		&& 0\leq \eta_m[n],\forall m, \forall n, \label{Cons:normalizing factor}\\
		&&\text{Constraints\ }  \eqref{Cons:A UAV}, \eqref{Cons:A cluster}, \eqref{Cons:A binary}. \nonumber
	\end{eqnarray}
\end{subequations}
Constraints \eqref{Cons:target MSE} describe the minimum instantaneous MSE requirement $\varepsilon_l$ for data aggregation from any cluster $l$ in each time slot.
To avoid trivial problems, it is assumed that $\varepsilon_l<1, \forall l$. 
Otherwise, a trivial solution of problem \eqref{Problem:original}, i.e, $p_k[n] =0, \forall k \in \mathcal{K}, \forall n \in \mathcal{N}$ and $\eta_{l,m}[n] =0, \forall l \in \mathcal{L}, \forall m \in \mathcal{M},\forall n \in \mathcal{N}$ achieves the certain  instantaneous MSE requirements for any given $\bm Q$ and $\bm A$.
Constraints \eqref{Cons:power} are the maximum total transmit power constraints of all devices. 
Constraints \eqref{Cons:Q position} represent the UAVs' initial and final locations constraints.
Constraints \eqref{Cons:Q speed} represent that the UAVs are constrained by the maximum speed  $V_{\max}$  in meter/second (m/s).
Constraints \eqref{Cons:Q collision} limit the minimum safe distance between UAVs for collision avoidance with $d_{\min}$ denoting the minimum inter-UAV distance in meters. 
Constraints \eqref{Cons:normalizing factor} are derived from \eqref{Neq:instantaneous MSE}.
It can be easily verified that for problem $\mathscr{P}$, if the obtained binary variables meet condition  $\sum_{m=1}^{M}a_{l,m}[n] = 0$, the power solution for all devices in cluster $l$ must satisfy $p_k[n]=0, \forall k \in \mathcal{K}_l$.

The main challenges in solving problem \eqref{Problem:original} arise from the following aspects.
First, due to the binary cluster scheduling and  association variables $\bm A$,  max-min problem \eqref{Problem:original} involves a discontinuous objective function and  integer constraints \eqref{Cons:target MSE}, \eqref{Cons:A cluster}, and \eqref{Cons:A UAV}.
Second, all continuous variables $\{\bm P, \bm \eta, \bm Q\}$ are coupled with binary variables $\bm A$ in constraints \eqref{Cons:target MSE}, which are non-convex.
In addition, collision avoidance constraints \eqref{Cons:Q collision} are also non-convex.
As such,  problem \eqref{Problem:original} is a nonconvex max-min problem with MINLP, which is generally challenging to be optimally solved.
As a compromise approach, in the following subsection, we transform problem \eqref{Problem:original} into an equivalent MINLP problem to facilitate the development of a low-complexity algorithm.

\subsection{Problem Reformulation}
To address the challenges of the max-min problem with the discontinuous objective function,  an auxiliary variable $ \Gamma = \mathop{\text{max}}\limits_{(l,n)}  \frac{{\sf{MSE}}_{l,n} }{\varepsilon_l}$ is defined as the maximum achievable AirComp MSE ratio among all clusters.
Furthermore, we denote the minimum amount of AirComp tasks among all clusters as $D\in \mathbb{Z}$.
Given any $D$, we first introduce the following optimization problem:
\begin{subequations}
	\begin{eqnarray}
		\mathscr{P}_1:\mathop{\text{minimize}}\limits_{\Gamma, \bm{A},\bm{P}, \bm{\eta},\bm{Q}}&& \Gamma(D)\nonumber\\
		\text{subject to } 
		&& \frac{{\sf{MSE}}_{l,n} }{\varepsilon_l}\le \Gamma(D), \forall n, \forall l, \label{Cons:target MSE1}\\
		&& \sum_{n=1}^N\sum_{m=1}^{M} a_{l,m}[n] \geq D, \forall l,\label{Cons:target amount}\\
			&& \text{Constraints\ } \eqref{Cons:A UAV}, \eqref{Cons:A cluster}, \eqref{Cons:A binary}, \eqref{Cons:power}, \eqref{Cons:normalizing factor},  \eqref{Cons:Q position}, \eqref{Cons:Q speed}, \eqref{Cons:Q collision}. \nonumber
	\end{eqnarray}
\end{subequations} 
Problem $\mathscr{P}_1$ aims to minimize the maximum ratio $\Gamma$,  which is a function of $D$.
By denoting the optimal objective value of $\mathscr{P}_1$ as $\Gamma^{\star}(D)$, we have the following proposition.
\begin{proposition} \label{Proposition:P1 feature}
	The optimal objective value $\Gamma^{\star}(D)$ of problem $\mathscr{P}_1$ is a non-decreasing of $D$.
\end{proposition}
\begin{proof}
	For any given task amount, i.e., $D$ and $D'$, such that $D'>D$, we should prove that $\Gamma^{\star}(D')\geq\Gamma^{\star}(D)$. 
	Given $D'$, the optimal solution to $\mathscr{P}_1$ is denoted as $\{\Gamma^{\star}, \bm{A}^{\star},\bm{P}^{\star}, \bm{\eta}^{\star},\bm{Q}^{\star}\}$. 
	As the minimum task amount decreases to $D$, a feasible solution to $\mathscr{P}_1$ with the given $D$ denoting as $\{ \bar{\bm{A}}, \bar{\bm{P}}, \bar{\bm{\eta}}, \bar{\bm{Q}}\}$ can be constructed by letting $\bar{\bm{Q}} = \bm{Q}^{\star}$, $\bar{\bm{\eta}} = \bm{\eta}^{\star}$, $\bar{\bm{A}} = \bm{A}^{\star}$, and $ \bar{\bm{P}} = \bm{P}^{\star}$. In this case, the objective value  $\Gamma^{\star}(D') = \bar{\Gamma}(D)$. 
	$\bar{\Gamma}(D)$ is the objective value for a feasible solution $\{ \bar{\bm{A}}, \bar{\bm{P}}, \bar{\bm{\eta}}, \bar{\bm{Q}}\}$ to $\mathscr{P}_1$, which is obviously no less than the objective function for the optimal solution $\Gamma^{\star}(D)$. Thus, we have $\Gamma^{\star}(D') = \bar{\Gamma}(D) \geq \Gamma^{\star}(D)$, which completes the proof.
\end{proof}

For any given $D$, the MSE requirements $\{\varepsilon_l\}$  in each cluster  are achievable if and only if $\Gamma^{\star}(D)\leq 1$. 
Thus, it can be verified that problem $\mathscr{P}$ is equivalent to
\begin{subequations}
	\begin{eqnarray}
		\mathscr{P}_2:\mathop{\text{maximize}}\limits_{D}&& D\nonumber\\
		\text{subject to } 
		&& \Gamma^{\star}(D)\leq 1. \label{Cons:Gamma}
	\end{eqnarray}
\end{subequations} 
With such a transformation, we can concentrate on addressing problem $\mathscr{P}_2$ in the remainder of the paper.
Fortunately, based on Proposition \ref{Proposition:P1 feature}, problem $\mathscr{P}_2$ can be efficiently solved by applying the bisection search over the minimum amount of AirComp tasks  $D$ until the equality in \eqref{Cons:Gamma} holds.
As a result, the main difficulty of solving problem $\mathscr{P}_2$ lies in finding an efficient algorithm to solve problem $\mathscr{P}_1$ with any given amount $D$.  
Compared to the original problem $\mathscr{P}$,  problem $\mathscr{P}_1$ eliminates the max-min programming and discontinuous objective function.
Although it still involves
coupled optimization variables and non-convex constraints, such as \eqref{Cons:A binary}, \eqref{Cons:target MSE1}, and \eqref{Cons:Q collision}, problem $\mathscr{P}_1$ promotes the development of an efficient suboptimal algorithm, which will be elaborated in section \ref{Section:algorithm}.

\begin{remark}\label{Remark:original problem}
The equations of all the constraints in \eqref{Cons:target amount} hold for the optimal solution of problem $\mathscr{P}_1$.
If any equality of constraints in \eqref{Cons:target amount} is not met, we can always reduce the value of $\sum_{n=1}^N\sum_{m=1}^{M} a_{l,m}[n]$ when other variables are fixed, while all other constraints are still satisfied, and the objective value of $\mathscr{P}_1$ remains unchanged. 
\end{remark}

\section{Proposed Algorithm}\label{Section:algorithm}
In this section,  we develop an efficient iterative algorithm to solve problem $\mathscr{P}_1$ based on the principle of BCD and SCA techniques. 
Specifically, to tackle the coupling between optimization variables, we adopt the well-known BCD technique \cite{xu2013block} to decouple problem $\mathscr{P}_1$ into four subproblems, which optimize cluster association $\bm A$, transmit power $\bm P$, normalizing factors $\bm \eta$, and UAVs' trajectories $\bm Q$ in an iterative manner.
Besides the convex transmit power $\bm P$ and normalizing factors $\bm \eta$ optimization subproblems, we propose a low-complexity dual ascent method to obtain a closed-form binary solution for  cluster association optimization subproblem and exploit the SCA to deal with the non-convexity of  UAVs' trajectories optimization subproblem.

\subsection{Cluster Scheduling and Association Optimization}
\begin{algorithm}[t]
	\caption{Dual Method for Problem \eqref{Subproblem:A original}}\label{Algo:A solution}
	\begin{algorithmic}[1]
		\STATE {\textbf{Input:} $M$, $N$, $L$, $D$,  $\{{\sf mse}_{l,m,n}/\varepsilon_l, \forall l,  \forall m, \forall n \}$}.
		\STATE {Initialize dual variables $\{\lambda_{l,n} = 1/(KN)\}$, $ \{\beta_{l,n} =  0\}$, and $\{\nu_{m,n} = 0\}$.}
     	\REPEAT
		\STATE {Update the primal variables $\bm A$ and $\bm \Gamma$ according to  \eqref{Eq:A-solution} and \eqref{Eq:Gamma-solution}}.
		\STATE {Update the dual variables $\bm \lambda$, $\bm \beta$, and $\bm \nu$ according to \eqref{Eq:beta}-\eqref{Eq:lambda}.}
		\UNTIL{$\bm \lambda$, $\bm \beta$, and $\bm \nu$ converge within a prescribed accuracy}	
			\STATE {Set $\bm a^{\star}$ = $\bm a$ and $ \Gamma^{\star}$ = $\Gamma$}.
		\STATE {\textbf{Output:} $\bm a^{\star}$ and $ \Gamma^{\star}$}	 		
	\end{algorithmic}
\end{algorithm}
For given  $\{ \bm \eta, \bm Q, \bm P  \}$,  the cluster scheduling and association optimization subproblem is reduced to as follows.
\begin{subequations}\label{Subproblem:A original}
\begin{eqnarray}
\mathop{\text{minimize}}\limits_{\Gamma, \bm A} && \Gamma \nonumber\\
\text{subject to} &&\sum_{m=1}^Ma_{l,m}[n]\frac{{\sf{mse}}_{l,m,n} }{\varepsilon_l}\le \Gamma, \forall l,  \forall n, \label{Cons:target MSE-A} \\
	&&\text{Constraints~} 	 \eqref{Cons:A UAV}, \eqref{Cons:A cluster}, \eqref{Cons:A binary}, \eqref{Cons:target amount}. \nonumber
\end{eqnarray}	
\end{subequations}	
Problem \eqref{Subproblem:A original} is a binary linear programming optimization problem.
Although such a problem can be directly solved via CVX and solver (e.g., Mosek),
it requires prohibitively high computational complexity due to exhausted search, which is not practical even for the moderate sizes of  $N$ and $L$.
Alternatively, to gain more design insights and reduce the computational cost, by relaxing integer constraints \eqref{Cons:A binary} with $ a_{l,m}[n]\in[0,1]$, we exploit the Lagrange duality to solve problem \eqref{Subproblem:A original} with relaxed constraints.
 It will show that the obtained closed-form optimal solution ensures not only the dual optimality but also the feasibility of problem \eqref{Subproblem:A original}. 
Specifically, the partial Lagrangian of problem \eqref{Subproblem:A original} with relaxed constraints is given by
\begin{eqnarray}
	\mathcal{L}(\eta, \bm A, \bm \lambda, \bm \beta, \bm \nu) && = \Big(1-\sum_{l=1}^L\sum_{n=1}^N\lambda_{l,n}\Big) \Gamma 
	+ \sum_{l=1}^L\sum_{n=1}^N \sum_{m=1}^Ma_{l,m}[n]\Big(\lambda_{l,n}\frac{{\sf{mse}}_{l,m,n} }{\varepsilon_l} +\beta_{l,n}+\nu_{m,n}\Big) \nonumber \\
&&	- \sum_{l=1}^L\sum_{n=1}^N\beta_{l,n} - \sum_{m=1}^M\sum_{n=1}^N\nu_{m,n},
\end{eqnarray}	
where $\bm \lambda = \{\lambda_{l,n}, \forall l, \forall n\}$,  $\bm \beta = \{\beta_{l,n}, \forall l, \forall n\}$,  and  $\bm \nu = \{\nu_{m,n}, \forall m, \forall n\}$ are the Lagrange multipliers associated with constraints \eqref{Cons:target MSE-A},  \eqref{Cons:A UAV}, and \eqref{Cons:A cluster}.
Accordingly, the Lagrange dual function of problem \eqref{Subproblem:A original} with relaxed constraints is expressed as
\begin{subequations}\label{Subproblem:A dual function}
	\begin{eqnarray}
		g (\bm \lambda, \bm \beta, \bm \nu) = &&\min\limits_{\Gamma, \bm A} 	\mathcal{L}(\eta, \bm A, \bm \lambda, \bm \beta, \bm \nu)  \nonumber\\
	&& \sum_{n=1}^N\sum_{m=1}^{M} a_{l,m}[n] \geq D, \forall l,\\
	&& 0\leq a_{l,m}[n]\leq 1, \forall l, \forall m, \forall n. \label{Cons:bineary relax}
	\end{eqnarray}	
\end{subequations}	
which leads to the following lemma.
\begin{lemma} \label{Lemma:A-Gamma feature}
To make $g (\bm \lambda, \bm \beta, \bm \nu) $ in problem \eqref{Subproblem:A dual function}   bounded, it must follows that $(1-\sum_{l=1}^L\sum_{n=1}^N\lambda_{l,n}) = 0$.
\end{lemma}
\begin{proof}
	This can be proved by contradiction. If $(1-\sum_{l=1}^L\sum_{n=1}^N\lambda_{l,n}) > 0$ or  $(1-\sum_{l=1}^L\sum_{n=1}^N\lambda_{l,n}) < 0$, then
	$g (\bm \lambda, \bm \beta, \bm \nu)\rightarrow -\infty$ by setting $\Gamma\rightarrow -\infty$ or $\Gamma\rightarrow \infty$. Therefore, both of the  two above inequalities cannot be true. This completes the proof.
\end{proof}
From Lemma \ref{Lemma:A-Gamma feature}, dual variables $\bm \lambda$ are subject to an additional constraint. 
Accordingly, the dual problem of \eqref{Subproblem:A original} with relaxed constraints is expressed as
\begin{subequations}\label{Subproblem:A dual problem}
	\begin{eqnarray}
	 \mathop{\text{maximize}}\limits_{\bm \lambda, \bm \beta, \bm \nu} &&	g (\bm \lambda, \bm \beta, \bm \nu)   \nonumber\\
	 \text{subject to}
		&& \sum_{l=1}^L\sum_{n=1}^N\lambda_{l,n} = 1, \bm \beta \succeq 0, \bm \nu \succeq 0.
	\end{eqnarray}	
\end{subequations}	
The next step is to apply the Lagrange duality to find the primal optimal solution.

\subsubsection{Obtaining $g (\bm \lambda, \bm \beta, \bm \nu) $}
First, we obtain the dual function $g (\bm \lambda, \bm \beta, \bm \nu) $ under given $\bm \lambda, \bm \beta, \bm \nu$ by solving problem \eqref{Subproblem:A dual function}.  
It is observed that we can decompose Problem \eqref{Subproblem:A dual function} into $L+1$ subproblems, each of which can be solved in parallel. 
Particularly, the $L$ subproblems are for optimizing $\bm A$, each of which is given by
\begin{subequations}\label{Subproblem:A dual function-A}
	\begin{eqnarray}
\mathop{\text{minimize}}\limits_{a_{l,m}[n], \forall m, \forall n} 	&&	\sum_{n=1}^N \sum_{m=1}^Ma_{l,m}[n]\Big(\lambda_{l,n}\frac{{\sf{mse}}_{l,m,n} }{\varepsilon_l} +\beta_{l,n}+\nu_{m,n}\Big) \nonumber\\
 \text{subject to}
		&& \sum_{n=1}^N\sum_{m=1}^{M} a_{l,m}[n] \geq D,\\
		&& 0\leq a_{l,m}[n]\leq 1,  \forall m, \forall n,
	\end{eqnarray}	
\end{subequations}	

To minimize the objective function in \eqref{Subproblem:A dual function-A},  which is a linear combination of $a_{l,m}[n]$, we should let the association coefficient corresponding to the UAV with the top $D$ smallest $\big(\lambda_{l,n}\frac{{\sf{mse}}_{l,m,n} }{\varepsilon_l} +\beta_{l,n}+\nu_{m,n}\big) $ be $1$ for any $l$. 
The optimal solution of problem \eqref{Subproblem:A dual function-A} is given by the following theorem.
\begin{theorem} \label{theorem:A-solution}
For problem \eqref{Subproblem:A dual function-A}, by denoting $x_{l,m}[n] = \lambda_{l,n}\frac{{\sf{mse}}_{l,m,n} }{\varepsilon_l} +\beta_{l,n}+\nu_{m,n}$, 
the optimal cluster scheduling and association $\bm A$ can be expressed as
\begin{eqnarray} \label{Eq:A-solution}
	&&\bm a_l^{\star}(\bm \pi_l(i))= 
	\left\{
	\begin{aligned}
		&1, i = 1,\ldots D, \\
		& 0, \  \text{otherwise},	
	\end{aligned}
	\right.
\end{eqnarray}
where vector $\bm a_l^{\star} \triangleq \big[a_{l,1}[1],\ldots, a_{l,1}[n], \ldots a_{l,1}[N],a_{l,2}[1],\ldots  a_{l,M}[N]\big]^{\sf T}\in \mathbb{R}^{MN\times 1}$,
vector $\bm x_l \triangleq \big[x_{l,1}[1],\ldots, x_{l,1}[n], \ldots x_{l,1}[N],\ldots,  x_{l,M}[N]\big]^{\sf T}\in \mathbb{R}^{MN\times 1}$, and permutation $\bm \pi_l = [\pi_l(1),\ldots, \pi_l(MN)]$ corresponds to the ascend order such that  $\bm x_l(\pi_l(1))\leq\cdots\leq\bm x_l(\pi_l(MN))$.
\end{theorem}
Theorem \ref{theorem:A-solution} states that cluster $l$ should be associated with UAV $m$ at time slot $n$ when coefficient $x_{l,m}[n]$ is smaller. 
According to the expression of coefficient $x_{l,m}[n]$, 
the term $\lambda_{l,n}\frac{{\sf{mse}}_{l,m,n} }{\varepsilon_l}$ is the effect of  introducing the ratio between the instantaneous MSE ${\sf{mse}}_{l,m,n}$ that can be achieved for type-$l$ AirComp task  and the target MSE threshold $\varepsilon_l$ that should be achieved if cluster $l$ is associated with UAV $m$ at time slot $n$. Furthermore, the remaining term represents other effects due to the problem constraints.
When $x_{l,m}[n]$ is smaller, it means that it brings a smaller achievable ratio if cluster $l$ is associated with UAV $m$ at time slot $n$.

And one subproblem is for optimizing $\Gamma$, i.e.,
\begin{eqnarray}\label{Subproblem:A dual function-Gamma}
	\mathop{\text{minimize}}\limits_{\Gamma} 	&&(1-\sum_{l=1}^L\sum_{n=1}^N\lambda_{l,n}) \Gamma 
\end{eqnarray}	
Since $\sum_{l=1}^L\sum_{n=1}^N\lambda_{l,n} = 1$, the optimal solution of $\Gamma^{\star}$ of problem \eqref{Subproblem:A dual function-Gamma} can be any arbitrary real number.
Without loss of generality, we simply set 
\begin{eqnarray} \label{Eq:Gamma-solution}
\Gamma^{\star} = \max\{a_{l,m}[n]\frac{{\sf{mse}}_{l,m,n} }{\varepsilon_l},\forall l, \forall n, \forall m\}.
\end{eqnarray}

\subsubsection{Updating $\bm \lambda, \bm \beta, \bm \nu $} 
After obtaining $\bm A^{\star}$ and $\Gamma^{\star}$ for given $\bm \lambda, \bm \beta$, and $\bm \nu $, 
the optimal dual variables are obtained by solvingthe dual problem \eqref{Subproblem:A dual problem}. 
Since the dual function $g(\bm \lambda, \bm \beta, \bm \nu)$  is concave but non-differentiable, 
subgradient-based methods such as the subgradient projection method are adopted to update the dual variables $(\bm \lambda, \bm \beta, \bm \nu)$. 
Specifically, in the $t+1$-th iteration, the update of  $(\bm \lambda, \bm \beta, \bm \nu)$ is given by
\begin{eqnarray}
&&\beta_{l,n}^{t+1} =  \Big[ \beta_{l,n}^{t}  + \gamma(\sum_{m=1}^Ma_{l,m}[n]-1) \Big]^+,  \nu_{m,n}^{t+1}=  \Big[ \nu_{m,n}^{t}  + \gamma(\sum_{L=1}^La_{l,m}[n]-1) \Big]^+,\label{Eq:beta}\\
&& \lambda_{l,n}^{t+\frac{1}{2}} = \Big[ \lambda_{l,n}^{t}  + \gamma(\sum_{m=1}^Ma_{l,m}[n] (\frac{{\sf{mse}}_{l,m,n} }{\varepsilon_l}-\Gamma) \Big]^+, \lambda_{l,n}^{t+1} =  \lambda_{l,n}^{t+\frac{1}{2}}/(\sum_{l=1}^L\sum_{n=1}^N\lambda_{l,n}^{t+\frac{1}{2}}), \label{Eq:lambda}\
\end{eqnarray}	
where $\gamma$ is a dynamically chosen step-size sequence and $[a]^+ = \max(a,0)$.

\subsubsection{Constructing the optimal solution $\Gamma$} 
With the updated dual variables $\bm \lambda^{\star}$, $\bm \beta^{\star}$, and $\bm \nu^{\star}$, we need to construct the primal solutions $\Gamma^{\star}$ and $\bm A^{\star}$ to problem \eqref{Subproblem:A original} with relaxed constraints. 
The key observation is that with given $\bm \lambda^{\star}$, $\bm \beta^{\star}$, and $\bm \nu^{\star}$, the optimal solution $\bm A^{\star}$ can be uniquely obtained from \eqref{Eq:A-solution}. 
By substituting $\bm A^{\star}$ into the primal problem \eqref{Subproblem:A original},  it is evident that  $\Gamma^{\star}  = \max\{a_{l,m}[n]\frac{{\sf{mse}}_{l,m,n} }{\varepsilon_l},\forall l, \forall n, \forall m\}$. 

By iteratively optimizing primal variables $(\bm A, \Gamma)$ and dual variables $(\bm \lambda, \bm \beta, \bm \nu)$,  the optimal cluster scheduling and association variables are obtained.
Algorithm \ref{Algo:A solution} summarizes how to solve problem \eqref{Subproblem:A original} via the dual method. 
It is important to note that the  obtained solution $a^{\star}_{l,m}[n]$ is either 1 or 0 according to \eqref{Eq:A-solution}, even though we relax $a_{l,m}[n]$ as \eqref{Cons:bineary relax}. 
Consequently, the optimal solution to problem \eqref{Subproblem:A original} is obtained by using the dual method.

\subsection{Transmit Power Optimization}   
By substituting $\bm A$ into constraints \eqref{Cons:target MSE1}, ${\sf MSE}_{l,n}$ is rewritten as
\begin{eqnarray} \label{Eq:instant MSE_A}
	&&{\sf {MSE}}_{l,n} = 
	\left\{
	\begin{aligned}
		&{\sf mse}_{l,m,n}, (l,m,n) \in \mathcal{A}, \\
		& 0, \  \text{otherwise},	
	\end{aligned}
	\right.	
\end{eqnarray}
where $\mathcal{A}\triangleq\{(l,m,n)|a_{l,m}[n] = 1, \forall l, \forall m, \forall n\}$. 
Thus, for any given  $\{ \bm \eta, \bm Q, \bm A \}$,  the transmit power optimization problem reduces to the following subproblem:	
\begin{subequations}\label{Subproblem:P solved}
\begin{eqnarray}
	\mathop{\text{minimize}}_{\Gamma, \bm P}&&\Gamma \nonumber\\
	\text{subject to} 
	&&\sum_{k\in\mathcal{K}_l}\!\!\Big({\theta_{l,k,m}[n]\sqrt{p_k[n]}}\!-\!1\Big)^2 \!\!+\sum_{j \neq l}\sum_{i\in\mathcal{K}_{j}}\theta_{l,i,m}^2[n]p_i[n] + \phi_{l,m}[n] \leq\Gamma \varepsilon_l K_l^2, 
	\forall (l,m,n) \!\in \!\mathcal {A}, \label{Cons:target MSE-P} \nonumber \\ \\
	&&\text{Constraints\ } \eqref{Cons:power}, \nonumber
\end{eqnarray}	
\end{subequations}
where $\theta_{l,k,m}[n] \triangleq \eta_{l,m}[n]|h_{k,m}[n]|, \forall l, \forall m, \forall k$ and $\phi_{l,m}[n] \triangleq \eta_{l,m}^2[n]\sigma^2, \forall l, \forall m, \forall k$.
Note that the number of constraints in \eqref{Cons:target MSE1} is $LN$, whereas the number of constraints in \eqref{Cons:target MSE-P} is $LD$.
Problem \eqref{Subproblem:P solved} is a quadratically constrained quadratic programming (QCQP), which can be solved via CVX with interior-point solvers (e.g., Mosek).
\begin{remark}\label{power remark}
	From constraints \eqref{Cons:target MSE-P}, to minimize $\Gamma$,  if there is no UAV associated with cluster $j$, i.e, $\sum_{m\in\mathcal{M}}a_{j,m}[n]=0$, then the transmission power of device $i\in\mathcal{K}_{j}$ must be zero (i.e., $p_i[n]=0$).
	Otherwise, it will introduce interference for UAVs and thus lead to a larger $\Gamma$.
\end{remark}

\subsection{Normalizing Factors Optimization} 
For given  $\{\bm A, \bm Q, \bm P \}$,  the normalizing factor optimization  subproblem is reduce to the following formulation:
\begin{subequations} \label{Subproblem:eta original}
	\begin{eqnarray}
		\mathop{\text{minimize}}_{\Gamma,\bm{\eta}} && \Gamma \nonumber\\
		\text{subject to}  
		&&\sum_{m=1}^M\frac{a_{l,m}[n]}{K_l^2}\big({\eta_{l,m}[n]\varphi_{k,m}[n]}-1\big)^2 +\eta_{l,m}^2[n]\big(I_{l,m}[n] + \sigma^2\big) \leq\Gamma \varepsilon_l K_l^2, \forall l, \forall n, \label{Cons:target MSE-eta}\\
		&&\text{Constraints\ } \eqref{Cons:normalizing factor},\nonumber
	\end{eqnarray}
\end{subequations}
where $\varphi_{k,m}[n] = |h_{k,m}[n]|\sqrt{p_k[n]}$.
Problem \eqref{Subproblem:eta original}  is also a convex QCQP. Furthermore, we can observe that problem \eqref{Subproblem:eta original} can be decoupled into $LD$ subproblems, each of which optimizes $\eta_{l,m}[n] $ with $a_{l,m}[n] = 1$ to minimize ${\sf{mse}_{l,m,n}}$. 
For any $(l,m,n)$ with $a_{l,m}[n] = 1$,  the subproblem is given by 
 \begin{eqnarray} \label{Subproblem:eta solved}
 	\mathop{\text{minimize}}\limits_{\eta_{l,m}[n]\ge0}  \sum_{k\in\mathcal{K}_l}\big({\eta_{l,m}[n]\varphi_{k,m}[n]}-1\big)^2 +\eta_{l,m}^2[n]\big(I_{l,m}[n] +\sigma^2\big).
\end{eqnarray}
And the optimal solution of problem \eqref{Subproblem:eta original} is given by  the following proposition.
\begin{proposition}
	Given by $\{ \bm A, \bm Q, \bm P   \} $, 
	by setting the first derivative of the objective function to zero,
	the optimal solution $\bm \eta$ to problem \eqref{Subproblem:eta solved}   is expressed as
\begin{eqnarray} \label{Solution:eta}
\eta_{l,m}^{\star}[n] = 
	\left\{
\begin{aligned} &\frac{\sum\limits_{k\in\mathcal{K}_{l}}\sqrt{p_k[n]}|h_{k,m}[n]|}{\sum\limits_{k\in\mathcal{K}_{l}}p_k[n]|h_{k,m}[n]|^2+ \sum_{j \in\mathcal{L}\backslash\{l\}}\sum_{i\in\mathcal{K}_{j}}p_i[n]|h_{i,m}[n]|^2+ \sigma^2}, \text{if}\ a_{l,m}[n] = 1, \\
	& 0, \  \text{otherwise}.		
\end{aligned}
\right.	
\end{eqnarray}
By substituting the solution $\eta_{l,m}^{\star}[n] $ into  problem  \eqref{Subproblem:eta original}, the optimal solution $\Gamma$ is given by
\begin{eqnarray} \label{Solution:eta-Gamma}
\Gamma^{\star} = \max\limits_{\forall l, \forall n}\sum_{m=1}^M\frac{a_{l,m}[n]}{K_l^2}\Big[\sum_{k\in\mathcal{K}_l}\Big(\eta_{l,m}[n]\sqrt{p_k[n]}|h_{k,m}[n]|-1\Big)^2 +\eta_{l,m}^2[n]\Big(I_{l,m}[n] + \sigma^2\Big)\Big].
\end{eqnarray}
\end{proposition}

\begin{remark}\label{eta remark}
	Note that  with  $a_{l,m}[n] = 1$, 
	the normalizing factor $\eta_{l,m}^{\star}[n] $ monotonically
	decreases with  respect to the  noise power $\sigma^2$ and the inter-cluster interference power from other devices associated with other UAVs, i.e., $\sum_{j \in\mathcal{L}\backslash\{l\}}\sum_{i\in\mathcal{K}_{j}}p_i[n]|h_{i,m}[n]|^2$.
	From \eqref{Solution:eta-Gamma}, we observe that reducing $\eta_{l,m}^{\star}[n]$ can suppress the noise-induced error and inter-cluster interference induced error components but increase the signal misalignment error.
\end{remark}

\subsection{UAV Trajectory Optimization}
For given  $\{ \bm A, \bm P, \bm\eta \}$,  the UAVs' trajectories subproblem  is reduced to the following subproblem:
\begin{eqnarray}\label{Subproblem:Q original}
	\mathop{\text{minimize}}\limits_{\Gamma,\bm{Q}} && \Gamma \nonumber\\
	\text{subject to}  
	&&\text{Constraints\ } \eqref{Cons:Q position},\eqref{Cons:Q speed}, \eqref{Cons:Q collision}, \eqref{Cons:target MSE1}. 
\end{eqnarray}	
Due to the nonconvexity of constraints \eqref{Cons:Q collision} and \eqref{Cons:target MSE1}, problem \eqref{Subproblem:Q original} is  nonconvex. 
Generally speaking, there is no efficient method that can be used to attain the optimal solution. 
To address their non-convexity, we exploit the SCA technique in the sequel. 
We first transform constraints \eqref{Cons:target MSE1} into a tractable form, which facilitates the development of the SCA technique.
We define
\begin{eqnarray} 
	\!\!\!\!\!\!\!\!\!\! G_{l,m}[n] &=& \eta^2_{l,m}[n]\sum_{l'\in \mathcal{L}}  \sum\limits_{k\in\mathcal{K}_{l'}}\frac{p_k[n]\beta_0}{(H^2+ \left\|\bm q_{m}[n]-\mathbf{w}_k\right\|_2^2)^{\frac{\gamma}{2}}},  \\
	\!\!\!\!\!\!\!\!\!\!    F_{l,m}[n] &=&
	2\eta_{l,m}[n]\sum\limits_{k\in\mathcal{K}_{l}}\dfrac{\sqrt{\beta_0 p_k[n]}}{(H^2 + \left\|\bm q_{m}[n]-\mathbf{w}_k\right\|_2^2)^{\frac{\gamma}{4}}},\\
	\!\!\!\!\!\!\!\!\!\!	C_{l,m}[n] &=&
	\left\{
	\begin{aligned}
		&K+ \eta^2_{l, m}[n]\sigma^2,~\text{if~}a_{l,m}[n] =1,  \\
		& 0, \  \text{otherwise}.		
	\end{aligned}
	\right.	
\end{eqnarray}
 Based on the above introduced functions, we substitute  $\bm A$ into ${\sf MSE}_{l,n}$ in constraints \eqref{Cons:target MSE1}, which is given by
\begin{eqnarray} \label{Eq:instant MSE_Q}
&&{\sf {MSE}}_{l,n} = 
	\left\{
	\begin{aligned}
		& G_{l,m}[n] + C_{l,m}[n] - F_{l,m}[n], (l,m,n) \in \mathcal{A}, \\
		& 0, \  \text{otherwise}.	
	\end{aligned}
	\right.	
\end{eqnarray}
With \eqref{Eq:instant MSE_Q}, constraints \eqref{Cons:target MSE1} are transformed into
\begin{eqnarray} \label{Cons:target MSE_Q}
	G_{l,m}[n] + C_{l,m}[n] - F_{l,m}[n]\leq \Gamma \varepsilon_l K_l^2,  \forall (l,m,n) \in \mathcal{A}.
\end{eqnarray}
Although $G_{l,m}[n]$ and $F_{l,m}[n]$  in constraints \eqref{Cons:target MSE_Q} are non-convex w.r.t $\bm{q}_m[n]$, they are convex w.r.t $\|\bm{q}_m[n]-\mathbf{w}_k\|_2^2$. 
Base on the key observation, we first introduce slack variables $\mathbf{S} = \{s_{k,m}[n]\triangleq\left\|\bm q_{m}[n]-\mathbf{w}_k\right\|_2^2\ |\ \forall k, \forall m, \forall n\}$ to tackle the non-convexity of $G_{l,m}[n]$ in constraints \eqref{Cons:target MSE_Q}.
The reformulated  problem \eqref{Subproblem:Q original} is given by
\begin{subequations}\label{Subproblem:Q reduced}
	\begin{eqnarray}
		\mathop{\text{minimize}}\limits_{\Gamma,\textbf{Q}} && \Gamma \nonumber \\
		\text{subject to } && 
		\!\!\! \hat{G}_{l,m}[n] + C_{l,m}[n] -F_{l,m}[n] \le \varepsilon_l \Gamma, \forall (l, m, n)\in \mathcal{A},  \label{Cons:target MSE-QS}\\
		&&  s_{k,m}[n]\le\left\|\bm q_{m}[n]-\mathbf{w}_k\right\|_2^2,\forall k, \forall m, \forall n, \label{Cons:S slack}\\
		&&\text{Constraints\ } \eqref{Cons:Q position},\eqref{Cons:Q speed}, \eqref{Cons:Q collision},  \nonumber
	\end{eqnarray}	
\end{subequations}
where 
\begin{eqnarray} 
	\hat{G}_{l,m}[n] &=& \eta^2_{l,m}[n]\sum_{l'\in \mathcal{L}}  \sum\limits_{k\in\mathcal{K}_{l'}}\frac{p_k[n]\beta_0}{(H^2+ s_{k,m}[n])^{\frac{\gamma}{2}}}.
\end{eqnarray} 
Note that for problem \eqref{Subproblem:Q reduced}, it can be easily verified that  all constraints in \eqref{Cons:S slack} can be met with equality.
Otherwise,  when other variables are fixed, we can increase the value of $\bm S$ to further decrease the value of $\{\hat{G}_{l,m}[n]\}$ without violating all constraints in problem \eqref{Subproblem:Q reduced}, thereby reducing the objective function.
However, problem \eqref{Subproblem:Q reduced} is still a nonconvex optimization problem because of the nonconvexity of constraints  \eqref{Cons:target MSE-QS}, \eqref{Cons:S slack}, and \eqref{Cons:Q collision}.
\begin{algorithm}[t]
	\caption{BCD-SCA Algorithm for $\mathscr{P}_1$}	\label{Algo:BCD-SCA}
	\begin{algorithmic}[1]
		\STATE {\textbf{Input}}: $D$, tolerance $\epsilon>0$.
		\STATE {{Initialize}}: $\bm P^{0}$, $\bm \eta^{0}$, and $\bm Q^{0}$, let $r = 0$.
		\REPEAT		
		\STATE  Given $\!\{\!\bm P^{r}, \bm \eta^{r}, \bm Q^{r}\!\}$, solve problem \eqref{Subproblem:A original} with Algo. \ref{Algo:A solution}, and denote the solution as $\bm A^{r+1}$.
		\STATE   Given $\!\{\!\bm A^{r+1}, \bm \eta^{r}, \bm Q^{r}\!\}$, solve convex problem \eqref{Subproblem:P solved}, and the  solution is denoted as $\bm P^{r+1}$.
		\STATE   Given $\!\{\!\bm A^{r+1}, \bm P^{r+1}, \bm Q^{r}\!\}$, solve problem \eqref{Subproblem:eta original} based on \eqref{Solution:eta}, denote the solution as $\bm \eta^{r+1}$.
		\STATE  Given $\!\{\!\bm A^{r+1}, \bm P^{r+1},\bm \eta^{r+1}\}$, solve problem \eqref{Subproblem:Q solved},  and denote the  solution as $\bm Q^{r+1}$.
		\STATE Set $r = r+1$.
		\UNTIL The fractional decrease of the objective value is below a  tolerance $\epsilon$.	
		\STATE {\textbf{Output}}:   $\Gamma$, $\bm A^{r}$, $\bm P^{r}$, $\bm \eta^{r}$, and $\bm Q^{r}$.
	\end{algorithmic}
	\label{subAlgo}
\end{algorithm}  	

Fortunately, $F_{l,m}[n]$ in constraints \eqref{Cons:target MSE-QS} is convex w.r.t $\|\bm{q}_m[n]-\mathbf{w}_k\|_2^2$.
And $\|\bm{q}_m[n]-\mathbf{w}_k\|_2^2$ in constraints \eqref{Cons:S slack} and
$\left\|\bm q_m[n] - \bm q_i[n] \right\|_2$ in constraints  \eqref{Cons:Q collision}
are convex w.r.t $\bm Q$.
It is important to recall that the first-order Taylor expansion of any function at any point severs as a lower bound   \cite{Boyd2004convex}.  
Thus, similarly to \cite{Wu2018MultiUAV}, the SCA technique is adopted to tackle their non-convexity, thus obtaining a suboptimal solution.
Specifically, defining  $\bm Q^r$ as the output of the $r$-th iteration,
we obtain that $F_{l,m}[n]$ is lower bounded by the following expression, i.e.,
\begin{eqnarray}\label{Geq:F}
	 F_{l,m}[n]&\ge& F^r_{l,m}[n] \!+\! \nabla_{\|\bm q_{m}[n]-\mathbf{w}_k\|^2}F_{l,m}[n] \big|{\bm q_{m}[n] 
	 = \bm q_{m}^r[n]}\big( \left\|\bm q_{m}[n]-\mathbf{w}_k\right\|_2^2 -\left\|\bm q_{m}^r[n]\!-\!\mathbf{w}_k\right\|_2^2\big) \nonumber\\
	&\triangleq& F^{\text{lb}}_{l,m}[n], 
\end{eqnarray}
where
\begin{eqnarray} 
	\!\!\!\!\!\!\!\!\!&&\nabla_{\|\bm q_{m}[n]-\mathbf{w}_k\|^2}F_{l,m}[n] \big|_{\bm q_{m}[n] = \bm q_{m}^r[n]}  =
	\sum\limits_{k\in\mathcal{K}_{l}}-\dfrac{{\eta_{l,m}[n]}\gamma\sqrt{p_k[n]\beta_0}}{2(H^2+\left\|\bm q_{m}^r[n]-\mathbf{w}_k\right\|_2^2)^{\frac{\gamma+4}{4}}}.
\end{eqnarray}
It is easily verified that $F^{\text{lb}}_{l,m}[n]$ is concave with regard to $\bm q_{m}[n]$. Hence, $\hat{G}_{l,m}[n] + C_{l,m}[n] -F^{\text{lb}}_{l,m}[n]$ is a convex function. 
Likewise, with the first-order Taylor expansion at  point $\bm q_{m}^r[n]$, $\left\|\bm q_{m}[n]-\mathbf{w}_k\right\|_2^2$ in constraints \eqref{Cons:S slack} are subject to
\begin{eqnarray}\label{Geq:S slack}
	\left\|\bm q_{m}[n]-\mathbf{w}_k\right\|_2^2 &\ge& \left\|\bm q^r_{m}[n]-\mathbf{w}_k\right\|_2^2 + 2({\bm  q}^r_{m}[n]-\mathbf{w}_k)^T(\bm q_{m}[n]-\bm q^r_{m}[n]) \triangleq D^{\text{lb}}_{k,m}[n].
\end{eqnarray}
Note that the lower bound function $D^{\text{lb}}_{k,m}[n]$ is a linear function w.r.t $\bm q_{m}[n]$.
Likewise,  $\left\|\bm q_m[n] -\bm q_i[n] \right\|_2^2$ in constraint \eqref{Cons:Q collision} is lower-bounded by
\begin{eqnarray}\label{Geq:Q collision}
	\!\!	\left\|\bm q_m[n] -\bm q_i[n] \right\|^2 &\ge&-\left\|\bm q^r_{m}[n]-\bm q_i^r[n]\right\|^2 + 2(\bm q^r_{m}[n] -\bm q_i^r[n])^T(\bm q_{m}[n]-\bm q_{i}[n])\triangleq d^{\text{lb}}_{m,i}[n].
\end{eqnarray}
Given point $\bm Q^r$, based on  expressions  \eqref{Geq:F}, \eqref{Geq:S slack}, and \eqref{Geq:Q collision},  problem \eqref{Subproblem:Q reduced} approximates as follows:
\begin{subequations}\label{Subproblem:Q solved}
	\begin{eqnarray}
		\mathop{\text{minimize}}\limits_{\Gamma,\textbf{Q}} && \Gamma \nonumber \\
		\text{subject to } && 
		\!\!  \hat{G}_{l,m}[n] \!+\! C_{l,m}[n] \!-\!F^{\text{lb}}_{l,m}[n]\!\le\! \varepsilon_l \Gamma,\!\forall (l,m, n)\in \mathcal{A}, \label{Cons:target MSE SCA}\\
		&& 0\leq s_{k,m}[n]\le D^{\text{lb}}_{k,m}[n],\forall k, \forall m, \forall n,\label{Cons:S slack SCA}\\
		&& d^{\text{lb}}_{m,i}[n] \geq d_{\text{min}}, \forall n, \forall m, m\neq i, \label{Cons:Q collision SCA}\\
		&&\text{Constraints~} \eqref{Cons:Q position}, \eqref{Cons:Q speed}.\nonumber
	\end{eqnarray}	
\end{subequations}	
Since constraints in \eqref{Cons:target MSE SCA}, \eqref{Cons:S slack SCA}, and \eqref{Cons:Q collision SCA} are jointly convex w.r.t $\bm Q$ and $\bm S$ now, problem \eqref{Subproblem:Q solved} is  convex,  which can be solved via modeling framework CVX and interior-point solvers (e.g., Mosek).
Based on the lower bounds in \eqref{Cons:target MSE SCA}, we can conclude that any feasible solution to problem \eqref{Subproblem:Q solved} is also feasible to problem \eqref{Subproblem:Q original}. 
Additionally,  $-F^{\rm lb}_{l,m}[n]$ is an upper bound of $-F_{l,m}[n]$.
Therefore, the optimal objective value of the approximate problem \eqref{Subproblem:Q solved} is generally an upper bound to problem \eqref{Subproblem:Q original}.

	   \begin{algorithm}[t]
	   	\caption{Overall Bisection Algorithm for $\mathscr{P}_2$}	\label{Algo:main}
	   	\begin{algorithmic}[1]
	   		\STATE {\textbf{Input}}: $D_{\min} = 0$, $D_{\max} = \lfloor\frac{TM}{\delta L}\rfloor$.
	   		\STATE {{Initialize}}: $\bm P^{0}$ , $\bm \eta^{0}$, and $\bm Q^{0}$, let $t =1$.
	   		\REPEAT		
	   		\STATE	{Update $D = \lceil  \frac{D_{\min}+ D_{\max}}{2}  \rceil$. Initialize $\bm P^{0}$ and $\bm \eta^{0}$.}
	   		\STATE Obtain  $\Gamma^t$, $\bm A^t$, $\bm P^t$, $\bm \eta^t$, and $\bm Q^t$ using Algorithm \ref{Algo:BCD-SCA}.
	   		\IF{$\Gamma^t\leq 1$}
	   		\STATE Set $D_{\min} = D$. And let $ D^{\star} = D$, $\bm A^{\star} = \bm A^t$, $\bm P^{\star} = \bm P^t$, $\bm \eta^{\star} = \bm \eta^t$, and $\bm Q^{\star} = \bm Q^t$. 
	   		\ELSE	
	   		\STATE Set $D_{\max} = D$.	
	   		\ENDIF
	   		\STATE Update $t = t+1$.
	   		\UNTIL $D_{\max} - D_{\min} \leq 1$.	
	   		\STATE {\textbf{Output}}:   $D^{\star}$, $\bm A^{\star}$, $\bm P^{\star}$ , $\bm \eta^{\star}$, and $\bm Q^{\star}$.
	   	\end{algorithmic}
	   \end{algorithm}

 \subsection{Proposed Overall  Algorithm}
 The BCD-SCA method for solving $\mathscr{P}_1$ is summarized in Algorithm \ref{Algo:BCD-SCA}, where the cluster association, transmit power, normalizing factors, and UAV trajectories are successively optimized while keeping the other variables fixed until convergence. 
 In addition, the obtained solution at current iteration will be applied to be the input of the next iteration.
 And, the computational complexity of solving problem \eqref{Subproblem:A original} is $\mathcal{O}(LMN)$.
 The complexity of solving problem \eqref{Subproblem:P solved} is $\mathcal{O}(K^{3.5}L^{1.5}D^{1.5})$.
 The complexity of solving problem \eqref{Subproblem:eta original} is $\mathcal{O}(KMN)$.
 The complexity of solving problem \eqref{Subproblem:Q original} is $\mathcal{O}(K^{1.5}M^{3.5}N^{3.5})$.
 According to the aforementioned complexity analysis of each subproblem, the overall complexity of Algorithm \ref{Algo:BCD-SCA} is $\mathcal{O}((K^{3.5}L^{1.5}D^{1.5}+K^{1.5}M^{3.5}N^{3.5})log(1/\epsilon))$.
 
 Recall that problem $\mathscr{P}_2$ can be efficiently solved by applying the bisection search over the minimum amount of AirComp tasks $D$, then the overall bisection algorithm for $\mathscr{P}_2$ is summarized in Algorithm \ref{Algo:main}.
 Although the obtained solution is generally suboptimal, we show  the effectiveness of our proposed algorithm in increasing the number of AirComp tasks among all clusters via numerical simulations, and compare it to other benchmarks in Section \ref{Section:simulation}.
 
 Since our proposed algorithm is an iterative algorithm, in the following, we present the initial procedures for transmit power,  normalizing factors, and trajectories.    
{\it{1) Transmit power initialization:}}
 The transmit powers of the ground devices are initialized by the equal transmit power in all time slots, i.e., $p_k[n] = P_{k}/D, \forall k, \forall n$.
 {\it{2) Trajectory initialization:}} The trajectory  is initialized by a
 simple circular trajectory scheme, which is detailed in \cite{Wu2018MultiUAV}.
 {\it{3) Normalizing factor initialization:}} Given the initial transmit power and trajectories, the initial normalizing factors can be obtained by computing \eqref{Solution:eta}.

\section{Numerical Results}\label{Section:simulation}
In this section, numerical results are provided to verify performance gain of the proposed design in terms of the max-min AirComp task amount and the effectiveness of Algorithm \ref{Algo:main}.
We consider $L=6$ clusters.
Therein, each cluster has $|K_\ell| = 20, \forall \ell$ ground devices, each of which randomly distributes in a circular area with a radius of $r = 50$ meters, where the centers of circles are set to 
$(x_{C_1}, y_{C_1}, z_{C_1}) =(100, 50, 0)$ meters, $(x_{C_2}, y_{C_2}, z_{C_2}) = (200, 200, 0)$ meters, $(x_{C_3}, y_{C_3}, z_{C_3}) =(-100, 100, 0)$ meters, $(x_{C_4}, y_{C_4}, z_{C_4}) =(-400, 150, 0)$ meters, $(x_{C_5}, y_{C_5}, z_{C_5}) =(-200, -200, 0)$ meters, and $(x_{C_6}, y_{C_6}, z_{C_6}) =(-250, -100, 0)$ meters.
The UAVs fly at a fixed altitude $H = 100$ meters to comply with the rule that all commercial UAVs should not fly over $400$ feet ($122$ meters) \cite{FAA2016UAV}. 
In addition, the minimum distance between any two UAVs is set as $d_{\text{min}} = 100$ meters \cite{Wu2018MultiUAV}. 
The maximum speed of all UAVs is assumed to be the same and set as $V_{\text{max}} = 30 $ m/s. 
Time step size $\delta$ is set as $0.5$ s, which is small enough to satisfy $\delta V_{\text{max}}\ll H$. 
The noise power at the receiver and the channel power gain at the reference distance of $d_0 = 1$ m are set as $\sigma^2 = -80$ dBm and $\beta_0 = -50$ dB, respectively.
The target AirComp MSE threshold for all clusters is set as $\varepsilon_l = 2\times 10^{-3}, \forall l$.
Let the threshold $\epsilon$ in Algorithm \ref{Algo:BCD-SCA} be $10^{-3}$.

\subsection{Multi-Cluster AirComp with a Single UAV}
In this subsection, to show the superiority of our proposed design and the effectiveness of Algorithm \ref{Algo:main} for solving the minimum amount of AirComp tasks maximization problem, we first consider a special case that there is only one UAV serving ground devices, i.e., $M=1$. 
In this case, the system is free of inter-cluster interference.

\begin{figure}[t]
	\centering
	\begin{minipage}{.48\textwidth}
		\centering
		\includegraphics[width=8cm,height=6cm]{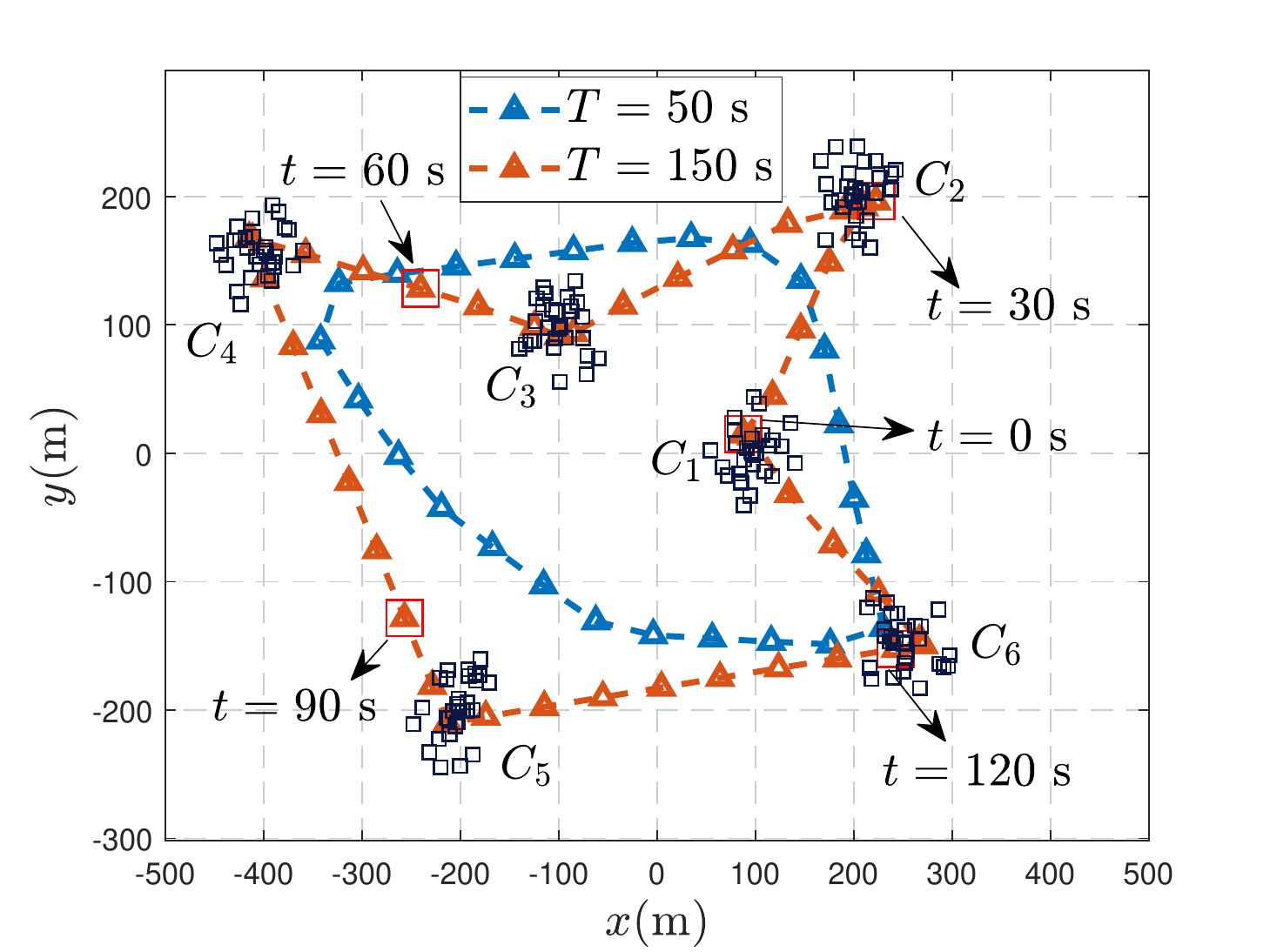}
		\vspace{-8mm}
		\caption{Optimized UAV trajectories in a single UAV case.}
	\label{Fig:traM1}
		\vspace{-6mm}
	\end{minipage}
	\hspace{4mm}
	\begin{minipage}{.48\textwidth}
		\centering
		\includegraphics[width=8cm,height=6cm]{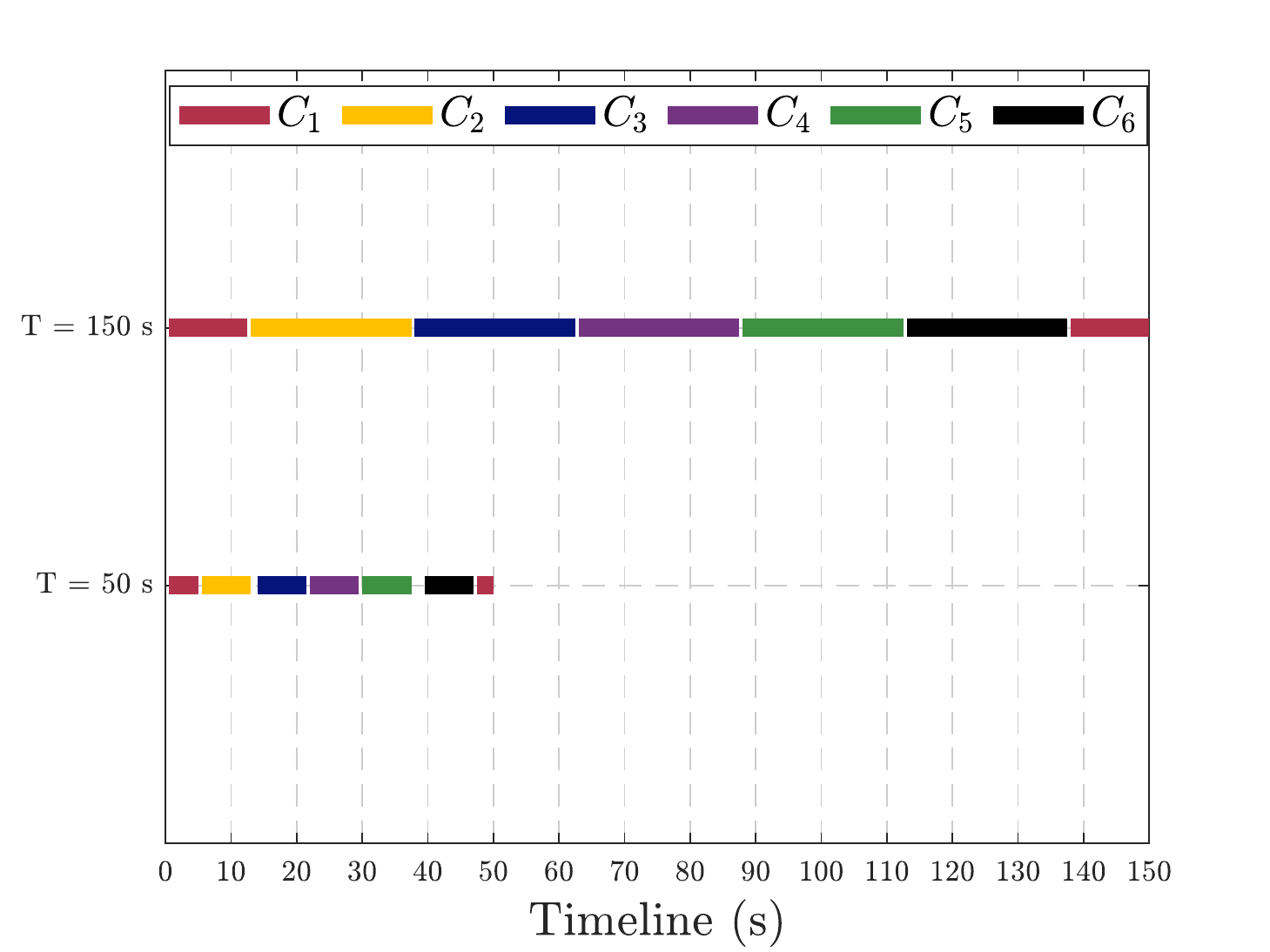}
		\vspace{-8mm}
		\caption{Cluster scheduling schemes over timeline.}
		\label{Fig:stateM1}
		\vspace{-6mm}
	\end{minipage}
\end{figure}

Fig. \ref{Fig:traM1} illustrates the obtained UAV trajectories optimized by using Algorithm \ref{Algo:main} with varying mission duration $T$ when $E=0.8$ W. 
Each trajectory is sampled every three seconds marked with ``$\triangle$'' by using the same colors.
The locations of devices are marked by dark blue ``$\square$".
We observe that, when $T$ is large, the UAV makes use of its maneuverability and adjusts its trajectory to approach its associated cluster as much as possible.
In particular, the UAV hovers over its associated cluster for a certain amount of time, except for the minimum time spent in flight between clusters.
Under the given mission duration $ T= 150$ s, the UAV can move sufficiently close to each cluster in sequence.
To this end, the optimized UAV trajectory contains line segments, which connect the points on the top of all clusters.

Fig. \ref{Fig:stateM1} shows the corresponding cluster scheduling along the timeline. 
The different colored rectangles represent different clusters, and their lengths represent how long these clusters are scheduled.
It is observed that for both $T = 50$ s and $T = 150$ s, the UAV visits each cluster in turn and associates with each cluster within an equal time period to complete the corresponding AirComp task under the target MSE requirements, which is in accordance with the observed results in Fig. \ref{Fig:traM1}.
%Compared to the case with $T = 50$ s,  in the case with $T = 150$ s, the achieved max-min AirComp task amount increases since a larger $T$ provides the UAV more time to fly closer to the associated clusters.
%Conversely, the access delay increases with $T$ since the cluster needs to wait longer for communication until the previous scheduled clusters finish.
%This result shows that there exists a fundamental delay-performance tradeoff.

To show the performance gain of our proposed joint design, we consider the following benchmarks.
\begin{itemize}
	\item \textbf{Static UAV}: This scheme fixes the location of UAV on the geometric point of all devices and only optimizes $\bm A$, $\bm \eta$, and $\bm P$ by using Algorithm \ref{Algo:main}.
	\item 	\textbf{Equal  power transmission}:   At each scheduled time slot, each device is allocated equal power (i.e., $p_k[n] = P/D$).
	Only three variables (i.e., $\bm A$, $\bm \eta$, and $\bm Q$) are optimized by using Algorithm \ref{Algo:main}.
	\item 	 \textbf{Upper bound}:  In this scheme, we assume that all devices have sufficient power budgets and the system is interference-free. Therefore, the maximum task amount for each cluster can be achieved, i.e.,
	\begin{eqnarray}\label{Eq:upper Bound}
		D^{\text{ub}} = \lfloor\frac{MT}{\delta L}\rfloor.
	\end{eqnarray}
\end{itemize}

\begin{figure}[t]
	\centering
	\begin{minipage}{.48\textwidth}
		\centering
		\includegraphics[width=8cm,height=6cm]{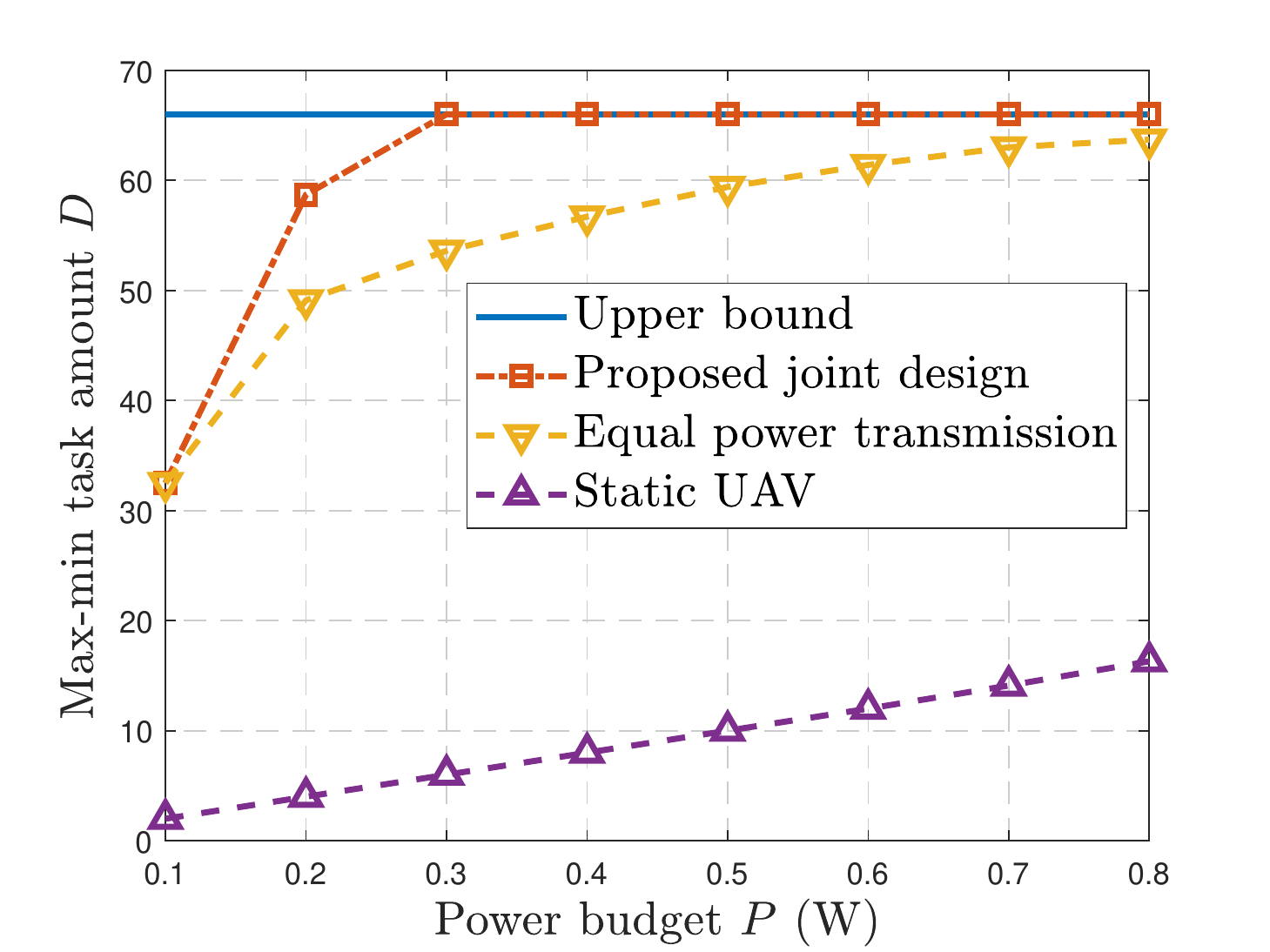}
		\vspace{-8mm}
		\caption{Max-min AirComp task amount versus power budget $P$ in a single UAV case.}
		\label{Fig:powerM1}
		\vspace{-6mm}
	\end{minipage}
	\hspace{4mm}
	\begin{minipage}{.48\textwidth}
		\centering
		\includegraphics[width=8cm,height=6cm]{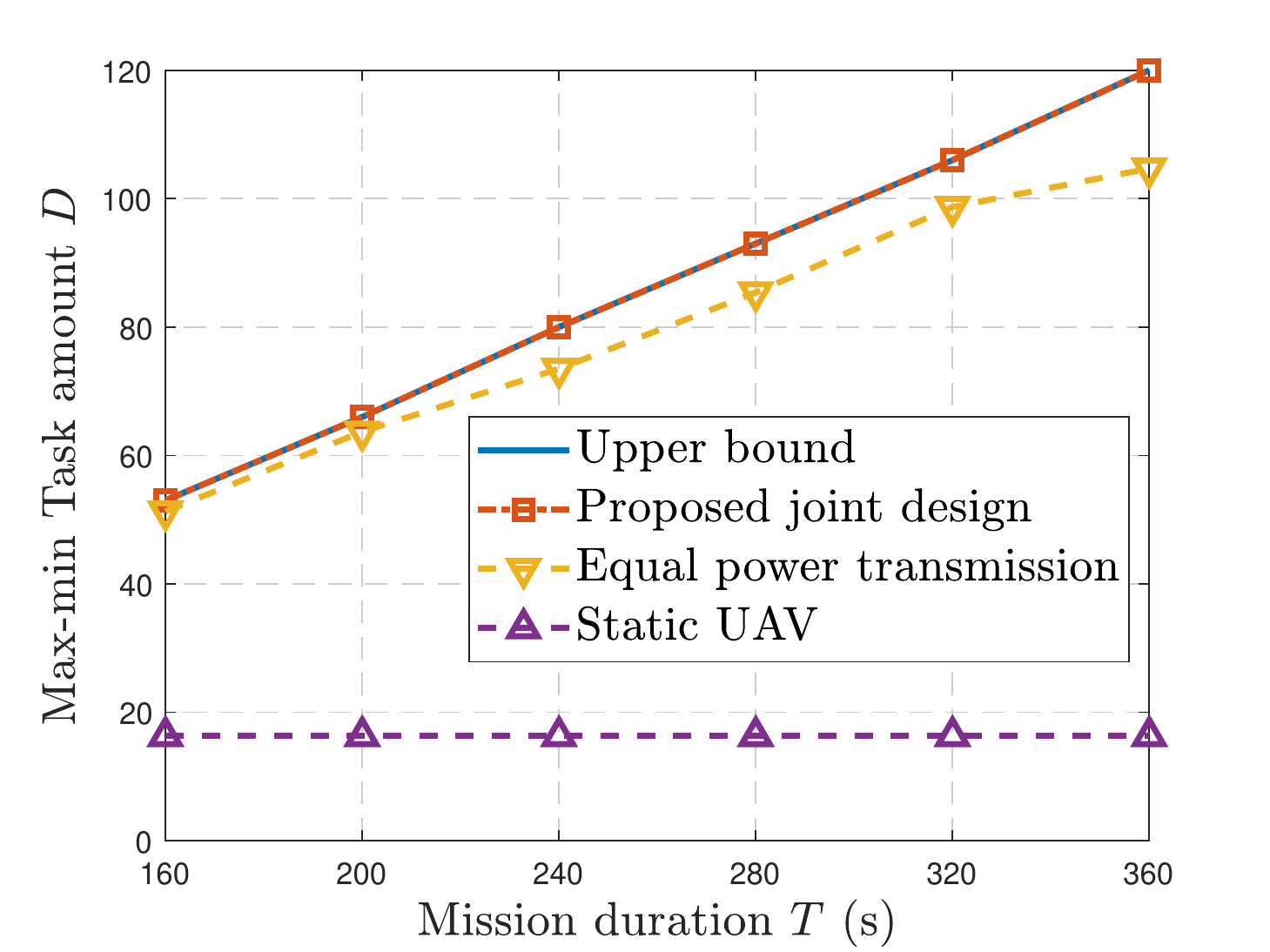}
		\vspace{-8mm}
		\caption{Max-min AirComp task amount versus mission duration $T$  in a single UAV case.}
		\label{Fig:timeM1}
		\vspace{-6mm}
	\end{minipage}
\end{figure}

Fig. \ref{Fig:powerM1} shows the average max-min task amount versus the total power budget when $T=200$ s. 
First, with the fixed mission duration, observing from \eqref{Eq:upper Bound}, the upper bound of max-min AirComp task amount is fixed.
With the increase of power budget $P$, the max-min task amount achieved by all three schemes increases.
This is because, with a larger $P$, the ground devices can afford more energy to compensate channel fading to reduce the achievable AirComp MSE.
However, as long-distance transmission may lead to deep fading, the performance gap between the static UAV scheme and upper bound is large even in the relatively large $P$ (i.e., $P = 0.8$ W).
Conversely, by utilizing UAV mobility, the proposed joint design scheme and the equal power transmission scheme significantly outperform the static UAV scheme. 
This is because the mobile UAV dynamically constructs favorable communication channels to avoid deep attenuation.
In addition, the performance of the proposed joint design is the best and reaches the upper bound when $P$ is relatively small (i.e., $P = 0.3$ W).
To achieve uniform powers of received signals, higher power is required to compensate for worse channel conditions. 
Using identical power transmission results in several misaligned signals since all devices' channel conditions change due to UAV mobility, especially when the power budget is small but ineffectively allocated.
Fortunately, the proposed joint design scheme exploits the synergy of trajectory design and power control to improve signal alignment.
All the above results illustrate the importance and necessity of the joint design in maximizing the max-min AirComp task amount.
This demonstrates the effectiveness of Algorithm \ref{Algo:main} for solving the joint design optimization problem $\mathscr{P}_2$.
It shows that in the single-UAV case without inter-cluster interference, the max-min task amount can achieve the upper bound by increasing the transmit power.

 Fig. \ref{Fig:timeM1} illustrates the average max-min task amount versus the mission duration $T$ when $P=0.8$ W.
 In the case of the static UAV, performance is independent of mission duration because of the time-invariant channel conditions between the UAV and devices.
 Conversely, with optimized trajectory design, the max-min AirComp task amount achieved by the proposed joint design scheme and the equal power transmission scheme increase with $T$ increasing. 
 There are two main reasons.
 First, the UAV with the optimized flight trajectory can establish good enough channel conditions with its associated cluster.
 Second, as $T$ increases, the UAV can spend more time performing tasks from its associated cluster under the established favorable channel conditions.
 We can further observe that the performance gap between the upper bound and the equal power transmission scheme increases as $T$ increases.
 However, thanks to effectively allocating power over all time slots, the proposed joint design scheme can reach the upper bound from $T = 160$ s to $T = 360$ s.
 Nevertheless, we can observe that there is a fundamental tradeoff between performance and delay.
 In the case of a single UAV, an increased max-min AirComp task amount means an increased access delay since the UAV sequentially visits each cluster and thus the cluster needs to wait longer for communication until the previous scheduled clusters to complete task amount.

%To sum up,   Fig. \ref{Fig:powerM1} and Fig. \ref{Fig:timeM1} show the effectiveness of Algorithm \ref{Algo:main} for solving the considered joint design optimization problem $\mathscr{P}$ and  the superiority of the proposed joint design scheme in increasing the max-min AirComp task amount compared with other benchmarks.
%Furthermore, according to the observed results in Fig. \ref{Fig:stateM1}, Fig. \ref{Fig:powerM1}, and Fig. \ref{Fig:timeM1}, in the case of a single UAV,  a sufficiently larger period $T$ can make full use of the devices' power to achieve performance improvement in the max-min task amount but at the cost of increasing device access delay.

\begin{figure}[t]
	\centering
	\begin{minipage}{.48\textwidth}
		\centering
		\includegraphics[width=8cm,height=6cm]{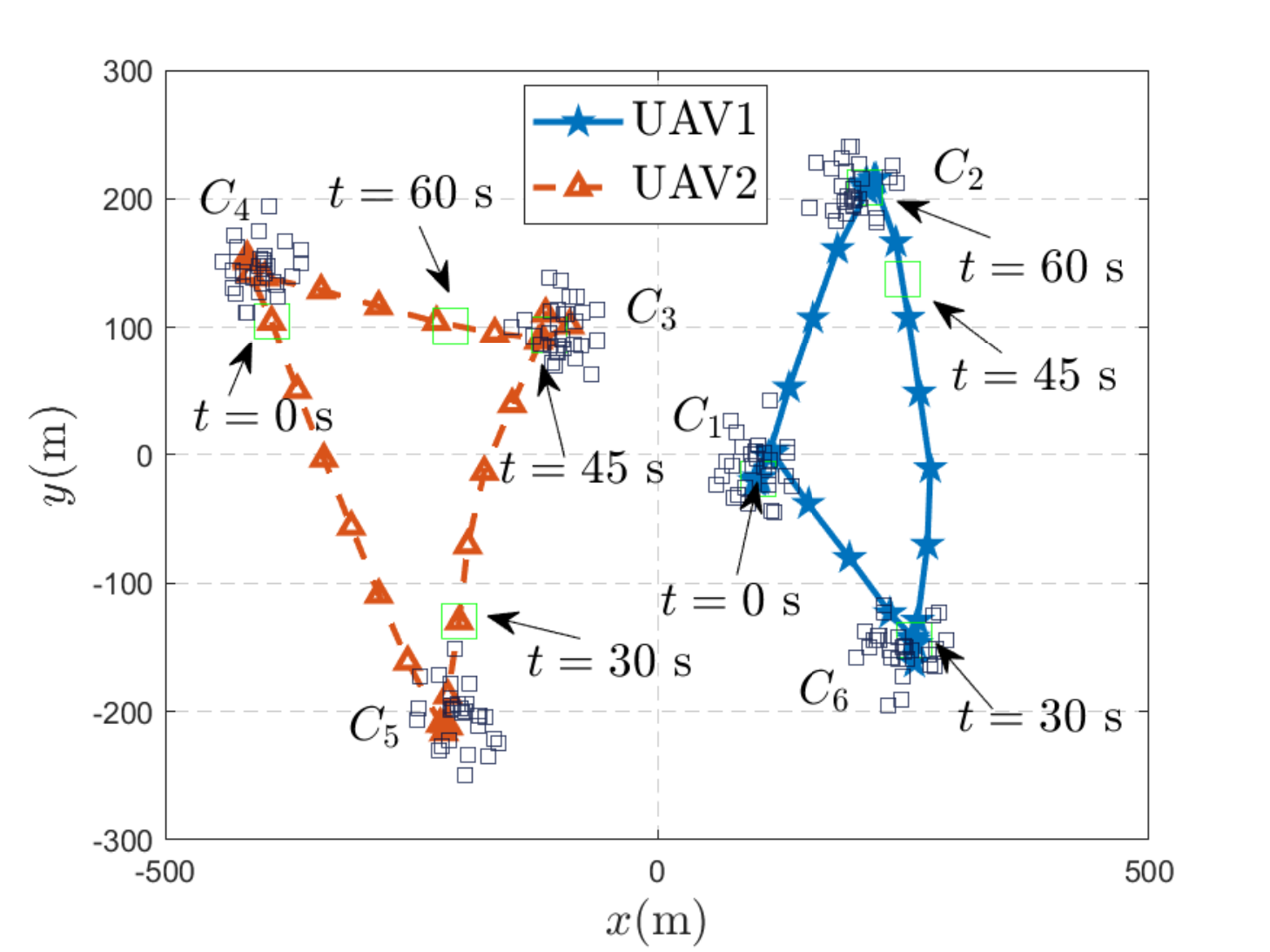}
		\vspace{-8mm}
		\caption{ UAV trajectories optimized by the proposed joint design in two-UAVs case.}
		\label{Fig:traM2}
		\vspace{-6mm}
	\end{minipage}
	\hspace{4mm}
	\begin{minipage}{.48\textwidth}
		\centering
		\includegraphics[width=8cm,height=6cm]{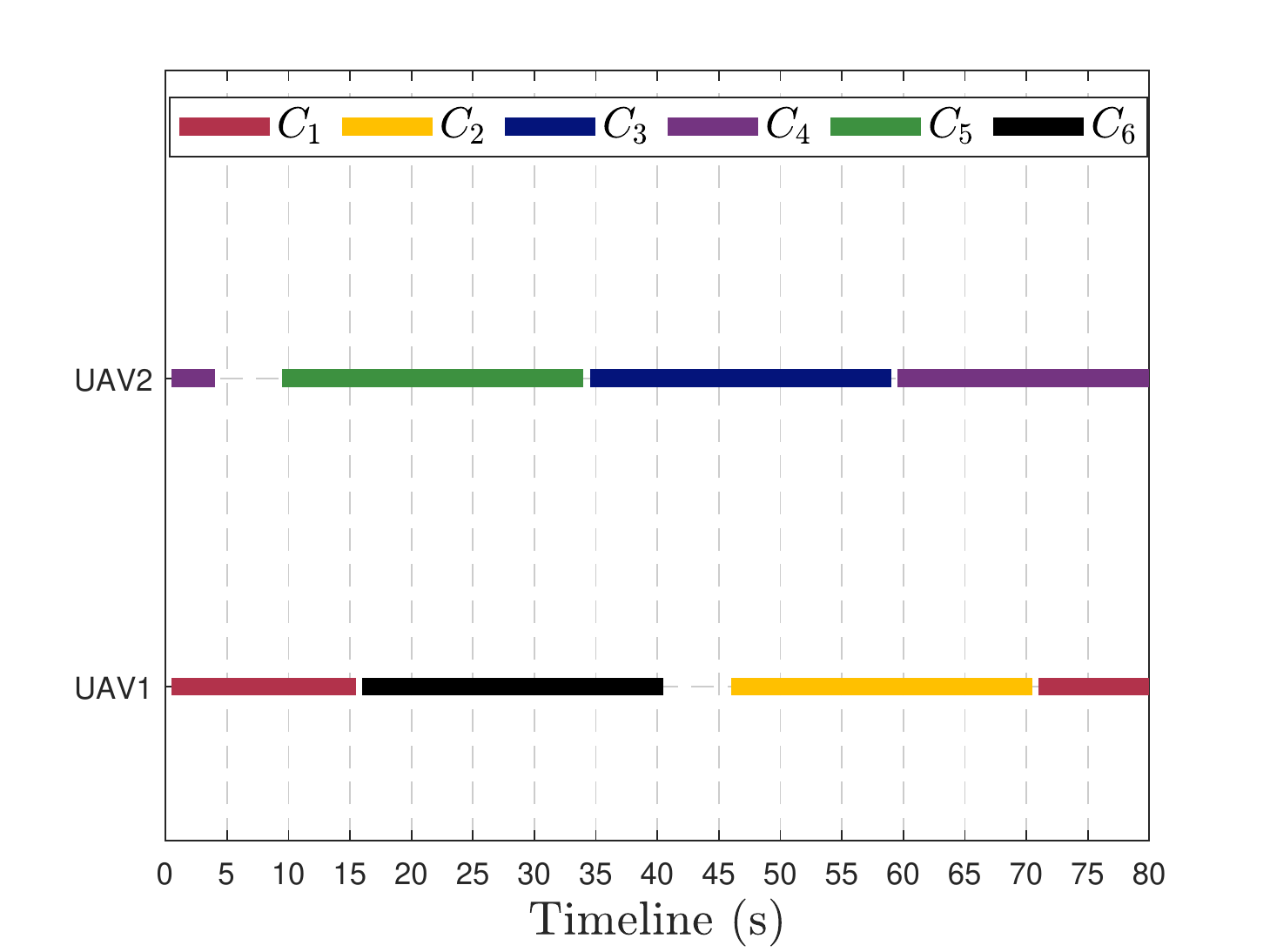}
		\vspace{-8mm}
		\caption{The corresponding cluster-UAV scheduling and association schemes over timeline  in two-UAVs case.}
		\label{Fig:stateM2}
		\vspace{-6mm}
	\end{minipage}
\end{figure}
\subsection{Multi-Cluster AirComp with Multiple UAVs}
Next, we study the multiple UAV case to illustrate its further performance gain.
To validate the effectiveness of our proposed cooperative interference management, here we take double UAVs as a simple example.
Additionally, our proposed algorithm is also valid for more than two UAVs.
We compare the performance of the proposed design to three benchmarks.
In the static UAV scheme, the UAVs are located at the geometric points of devices in clusters 1, 2, 6 and 3, 4, 5, respectively.
The following orthogonal UAV transmission scheme is included for comparison.
\begin{itemize}
	\item \textbf{Orthogonal transmission}: 
	The UAVs provide services to clusters via TDMA without inter-cluster interference. 
	In this scheme, the binary variables are subject to
	\begin{eqnarray}
		&&\sum\limits_{l=1}^L\sum\limits_{m=1}^M a_{l,m}[n] \leq 1,   n\in\mathcal{N}; a_{l,m}[n]\in\{0,1\}, \forall l\in\mathcal{L},m\in\mathcal{M},n\in\mathcal{N}. 
	\end{eqnarray}	
	This problem is a special instance of problem \eqref{Problem:original}, which can be solved  by using a similar Algorithm \ref{Algo:main}.
\end{itemize}

Fig. \ref{Fig:traM2} shows the trajectories of double UAVs optimized by the proposed design when $T = 80$ s and $P = 0.8$ W. 
To better observe phenomena, we also plot the corresponding cluster scheduling and association schemes along the timeline in Fig. \ref{Fig:stateM2}.
From Fig. \ref{Fig:traM2} and Fig. \ref{Fig:stateM2}, it can be observed that each cluster is served by the same UAV and each UAV only needs to be responsible for half of clusters. 
From Fig. \ref{Fig:stateM2}, we can observe that there is at least one UAV to complete one AirComp task at each time slot.
Additionally, at some time intervals (e.g., from $t=0$ s to $t=3$ s and from $t=10$ s to $t = 30$ s), there are two scheduled clusters, which means two UAVs actually work simultaneously to deal with different AirComp tasks. 
Furthermore, observed from Fig. \ref{Fig:traM2} and Fig. \ref{Fig:stateM2}, with the cluster-UAV scheduling and association design, to avoid strong co-channel interference, in most time slots, UAVs tend to choose to serve two clusters that are geographically far away from each other.
As a result, the two UAVs can properly enlarge their trajectories to move closer to their associated cluster for better channel conditions establishment while achieving weak co-channel interference for the other cluster. 
For example, as shown in Fig. \ref{Fig:stateM2}, from $t=70$ s to $t = 80$ s (i.e., $t = 0$ s), UAV $2$ associates with cluster $C_4$ while UAV $1$ associates with cluster $C_1$.
During this time interval, from Fig. \ref{Fig:traM2}, we can observe that UAV $2$ and UAV $1$ are respectively closer to cluster $C_4$ and cluster $C_1$.
In order to complete as many tasks as possible for each cluster, two UAVs will at some point have to serve two clusters simultaneously when flying.
For example, as shown in Fig. \ref{Fig:stateM2}, from $t = 35$ s to $t = 40$ s, UAV $2$ associates with cluster $C_3$ while UAV $1$ associates with cluster $C_6$.
In these time instants, the optimized trajectories show that UAVs tend to stay away from each other to reduce co-channel inter-cluster interference.
This is why the UAVs fly along arc paths with opposite opening directions during that time interval.

\begin{figure}[t]
	\centering
	\begin{minipage}{.48\textwidth}
		\centering
		\includegraphics[width=8cm,height=6cm]{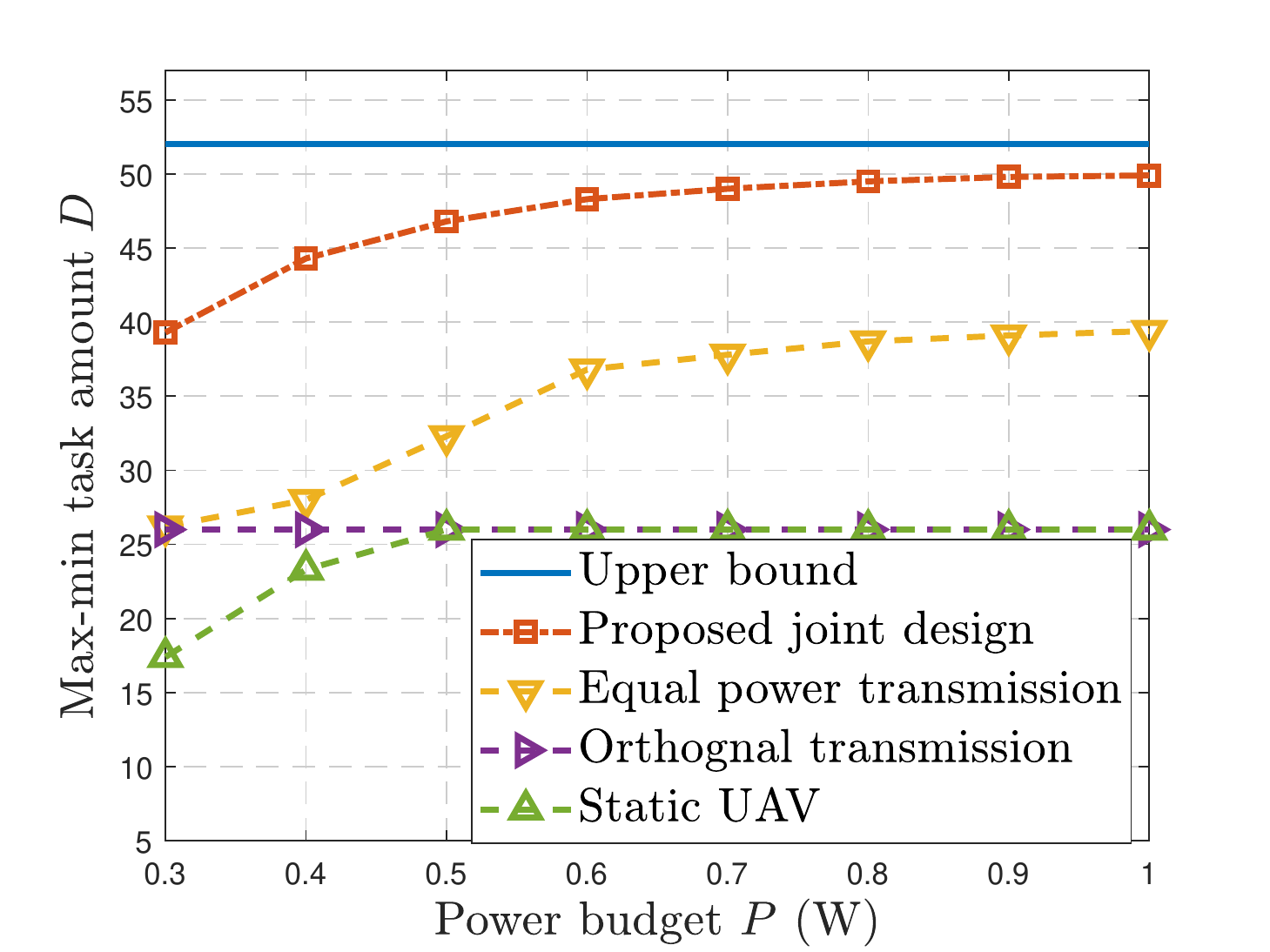}
		\vspace{-8mm}
		\caption{Max-min task amount versus power budget $P$ in two-UAVs case.}
		\label{Fig:powerM2}
		\vspace{-6mm}
	\end{minipage}
	\hspace{4mm}
	\begin{minipage}{.48\textwidth}
		\centering
		\includegraphics[width=8cm,height=6cm]{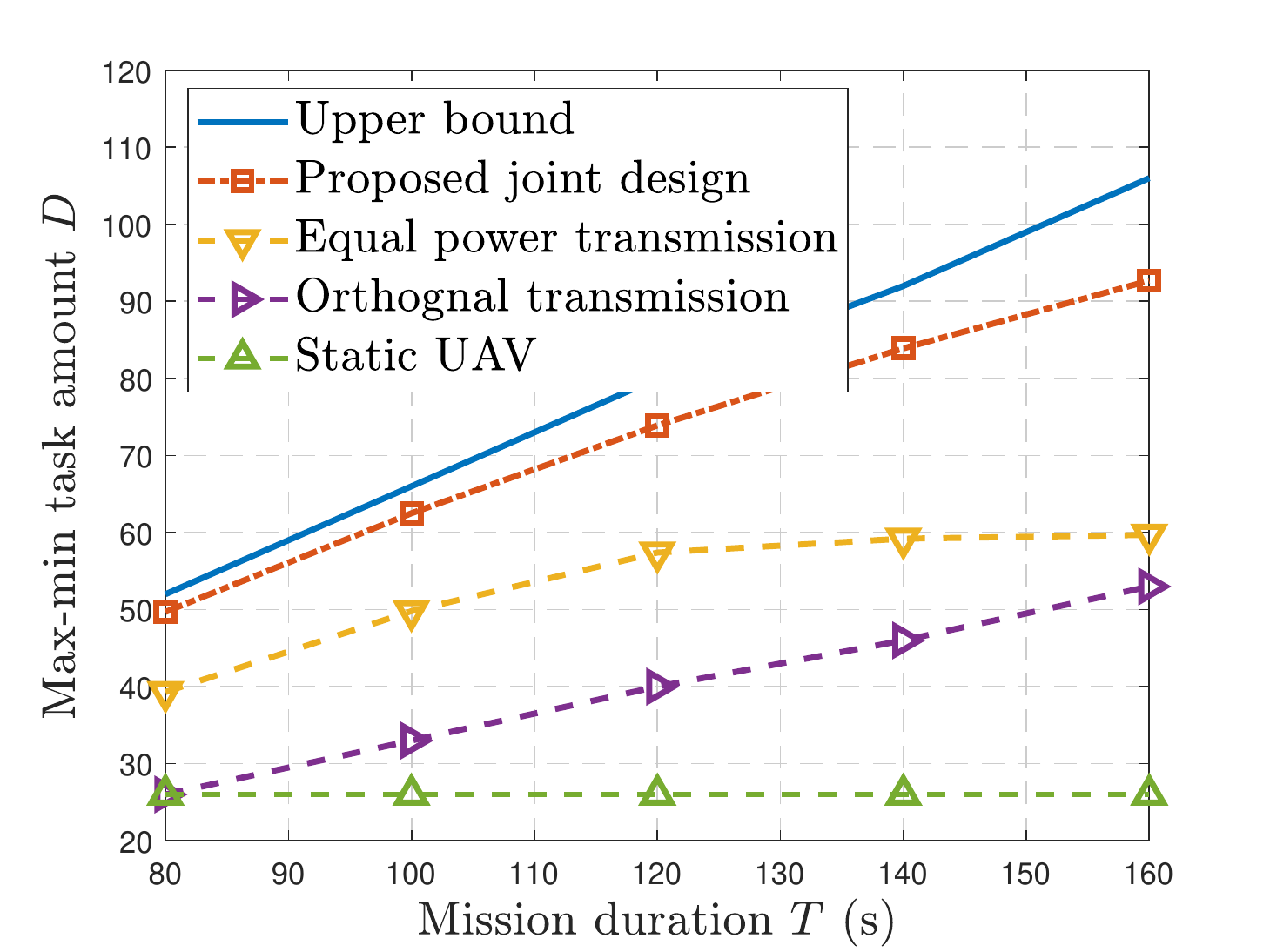}
		\vspace{-8mm}
		\caption{Max-min task amount versus mission duration $T$  in two-UAVs case.}
		\label{Fig:timeM2}
		\vspace{-6mm}
	\end{minipage}
\end{figure}

Fig. \ref{Fig:powerM2} shows the average max-min task amount versus device power budget when $T=80$ s. 
First, due to effective interference management, the proposed joint design scheme outperforms all benchmarks with all regimes considered.
Second, at the low power budget regime, the performance gaps between the proposed joint design and other benchmarks are relatively large.
This implies the effectiveness of power control optimization in suppressing the noise-induced error that is dominant for the achievable MSE in the low power budget regime.
Moreover, the proposed joint design and the equal-power transmission scheme outperform the orthogonal transmission scheme, demonstrating the benefit of cooperative transmission over clusters.
And their performance gaps become large as the power budget increases.
In addition, it is observed that all schemes suffering from inter-cluster interference become saturated in the high power budget regime, meaning that the achievable MSE performance is limited by the inter-cell interference and cannot be reduced further by simply increasing the transmit power.

Fig. \ref{Fig:timeM2} shows the average max-min task amount versus the mission duration $T$ when $P=0.8$ W in the case of two UAVs.
It can be observed that the performance of static UAV is the worst due to long distance transmission.
%Similar to Fig. \ref{Fig:timeM1}, the max-min AirComp task amount achieved by the static UAV is still regardless of the duration $T$. 
%The different point is that the achieved max-min task amount in two static UAVs is larger than that in a single static UAV.
%This is because in the two-UAV case, the communication distances between each static UAV and served devices reduce.
%As $T$ increases, with trajectory design, the max-min task amounts increase, which further demonstrates the fundamental performance-access delay tradeoff.
Compared to the curves of the proposed joint design scheme in Fig. \ref{Fig:timeM1}, it can be observed that with the same given $T$, the deployment of more than one UAV with cooperative interference management achieves a larger max-min task amount, thereby improving the performance-delay tradeoff.
For example,  when $T=160$ s, max-min AirComp task amount achieved by the proposed joint design in the single-UAV system is 53, whereas that in the two-UAV system exceeds 92.

\section{Conclusion}\label{Section:conclusion}

We study the minimum task amount maximization problem in a UAV-aided AirComp system with multiple clusters, taking into account cluster scheduling and association, the UAV trajectory design, and AirComp transceiver design.
The joint design aims to align signals within each cluster while mitigating the detrimental inter-cluster interference, thereby improving WDA performance.
By applying the Bisection technique, the formulated mixed-integer non-convex optimization problem is reduced to a sequence of the maximum ratio minimization problems.
This resulting problem is further solved by developing an efficient algorithm by applying BCD, Lagrange dual ascent, and SCA techniques.
Simulation results demonstrate significant performance gains of our proposed joint design over the benchmark schemes. 
%In addition, compared to the single-UAV case and OMA transmission, it shows that the use of multiple UAVs with effective cooperative interference management can considerably improve system performance while reducing access delay.
%For future studies, the joint design framework developed in this paper can be extended to more general scenarios with multiple cooperative UAVs and in the presence of ground BSs, while taking into account the prediction errors on the movement of devices for practical implementation.

%\appendix 
%\vspace{-1mm}
%\subsection{Proof of Proposition \ref{Proposition:P1 feature}}\label{P1 feature proof}

%\subsection{Proof of Lemma \ref{Lemma:A-Gamma feature}}\label{A-Gamma feature proof}

\bibliographystyle{IEEEtran}
\bibliography{ref} % BM
\end{document}